\documentclass[12pt,peerreview, letterpaper, onecolumn]{IEEEtran}
\usepackage{amsfonts,amsmath,amssymb}
\usepackage{cite}
\usepackage{graphicx}
\usepackage{url}
\newtheorem{theorem}{Theorem}
 
 \newtheorem{lemma}{Lemma}
  \newtheorem{proposition}{Proposition}
 \newtheorem{corollary}{Corollary}
 \newtheorem{definition}{Definition}
  \newtheorem{remark}{Remark}

\begin{document}

\baselineskip 4.6ex

\title{ High-Rate and Full-Diversity  Space-Time Block Codes with Low Complexity Partial Interference Cancellation Group Decoding}
 \author{
Long~Shi,~\IEEEmembership{Student Member,~IEEE,} ~Wei~Zhang,~\IEEEmembership{Member,~IEEE,} and ~Xiang-Gen~Xia,~\IEEEmembership{Fellow,~IEEE}
 \thanks{Manuscript was submitted to  IEEE Trans. Commun.  on 23 March, 2010.}
\thanks{L. Shi and W.  Zhang are with School of Electrical Engineering and Telecommunications,   University of New South Wales,  Sydney, Australia (e-mail: \{long.shi; w.zhang\}@unsw.edu.au). Their work was supported in part by the Australian Research Council Discovery Project DP1094194.}
\thanks{X.-G. Xia is with Department of Electrical and Computer Engineering, University of Delaware, DE 19716, USA (e-mail: xxia@ee.udel.edu). His work was supported in part by the Air Force Office of Scientific
Research (AFOSR) under Grant No. FA9550-08-1-0219 and the World Class University (WCU) Program 2008-000-20014-0.}
}
\maketitle
\begin{abstract}
In this paper, we propose a systematic design of space-time block codes (STBC) which  can achieve high rate and full diversity when the partial interference cancellation (PIC) group decoding is used at receivers. The proposed codes can be applied to any number of transmit antennas and admit a low decoding complexity while achieving full diversity. For $M$ transmit antennas, in each codeword real and imaginary parts of $PM$ complex information symbols are parsed into $P$ diagonal layers and then encoded, respectively. With PIC group decoding, it is shown that the decoding complexity can be reduced to a joint decoding of $M/2$ real symbols.  In particular, for $4$ transmit antennas, the code has  real symbol pairwise (i.e., single complex symbol) decoding that achieves full diversity
 and the code rate is $4/3$.  Simulation results demonstrate that the full diversity is offered by the newly proposed STBC with the PIC group decoding.
\end{abstract}

 \begin{keywords}
 MIMO systems, Space-time block codes, partial interference cancellation, decoding complexity
 \end{keywords}

\newpage

\section{Introduction}
Full diversity and low decoding complexity have been considered as two fundamental properties which a good space-time block code (STBC) should possess for multiple-input multiple-output (MIMO) wireless communications. The first orthogonal STBC (OSTBC) was proposed by Alamouti which can achieve full transmit diversity for two transmit antennas\cite{Alamouti}. Inspired by the Alamouti scheme, seminal studies focused on the designs of OSTBC for its unique orthogonal  code structure which  ensures a single symbol  maximum likelihood (ML) decoding  \cite{Tarokh98}\cite{Tarokh00}\cite{Lu}. However, OSTBC suffers from the reduced symbol rate with an increase of the number of transmit antennas, especially when complex constellations are used \cite{Wang}. In spite of full diversity advantage, OSTBC fails to achieve full channel capacity in MIMO channels \cite{sandhu}.
To address the problem of low symbol rate and capacity loss  in OSTBC,  linear dispersion code (LDC) was proposed as a full-diversity scheme that is constructed linearly in space and time \cite{has} \cite{heath}. The LDC design can be viewed as a linear combination of a fixed set of dispersion matrices with the transmitted symbols (or equivalently, combining coefficients). Diagonal algebraic space-time (DAST) block codes in \cite {dast} and threaded algebraic space-time (TAST) codes in \cite{tast} were also proposed as two typical algebraic designs which can obtain both full diversity and full rate with moderate ML decoding complexity. However, it is noted that the  aforementioned high-rate codes  rely on ML decoding to collect full diversity which has high decoding complexity. Efficient designs of STBC with low decoding complexity were proposed, such as  coordinate interleaved orthogonal design (CIOD) with single-symbol ML decoding  in \cite{Ra} and quasi-orthogonal STBC (QOSTBC) in \cite{Yuen}\cite{quan}, but their rates are restricted by the rates of OSTBC.

 Recently, full diversity achieving  STBC
 based on linear receivers, such as the minimum mean square error (MMSE) receiver and zero-forcing (ZF) receiver,  were studied and proposed \cite{Liu}, \cite{Shang}. However, it was shown in \cite{Shang} that the rates of these STBC based on linear receivers are not more than one. To address the complexity and rate tradeoff, a general decoding scheme with code design criterion, referred to as  partial interference cancellation (PIC) group decoding algorithm, was proposed in \cite{guoxia}. In the PIC group decoding, the symbols to be decoded are divided into several groups after a linear PIC operation and then each group is decoded separately with ML decoding.  Therefore, the PIC group decoding can be viewed as an intermediate decoding between ML decoding and ZF decoding. Apparently, PIC group decoding complexity depends on the number of symbols to be decoded in each group. Moreover, a  successive interference cancellation (SIC)-aided PIC group decoding was proposed in \cite{guoxia}. Based on the design criterion of STBC with PIC group decoding derived in \cite{guoxia}, a systematic design of STBC achieving full diversity under PIC group decoding was developed in \cite{wei}. In subsequent work, a new design of STBC having an Alamouti-Toeplitz structure was proposed in \cite{jisit} which provides a lower PIC decoding complexity compared with the design in \cite{wei}. However, the decoding complexity of the STBC in \cite{jisit} is equivalent to a joint decoding of $M/2$ complex symbols.

In this paper, we propose a design of STBC with PIC group decoding that can achieve both full diversity and low decoding complexity. The decoding complexity is equal to a joint decoding of $M/2$ \emph{real} symbols  for $M$   transmit antennas, i.e., only half decoding complexity of the STBC in \cite{jisit}. For the proposed STBC,  real and imaginary parts of $PM$ complex information symbols are parsed into $P$ diagonal layers and encoded by linear transform matrices, respectively. The full diversity can be achieved by the proposed STBC with $P=2$ under PIC group decoding and with any  $P$ under  PIC-SIC group decoding, respectively.
 The code rate is equal to $\frac{PM}{2P+M-2}$.
In particular, for $4$ transmit antennas   the
code has real symbol pairwise (i.e. single complex symbol) decoding.
Furthermore, the code rate is $4/3$.
It should be noted that the existing
STBC with single complex symbol (or real symbol pairwise) decoding,
such as QOSTBC  \cite{Yuen},  \cite{quan}, and CIOD etc.
 in \cite{Ra} have symbol rates not larger than one.
Also the codes with linear receivers have single complex symbol decoding
but their rates can not be above one either \cite{Shang}.
 Simulation results show that the proposed code outperforms the CIOD in \cite{Ra} and the QOSTBC with the optimal rotation in \cite{quan} for $4$ transmit antennas at the same bandwidth efficiency.  Moreover, our code guarantees full diversity without performance loss compared with other PIC group decoding based STBC in \cite{guoxia}, \cite{wei} and \cite{jisit}, but a half decoding complexity is reduced.
It should be mentioned that the major difference between the code in \cite{jisit} and the one proposed in this paper is that a complex-valued linear transform matrix is used for input complex signal vector to construct the code in \cite{jisit}, whereas in this paper   two real-valued linear transform matrices  are used for real and imaginary parts of the signals, respectively. 
By doing so, half decoding complexity can be reduced.

 The rest of this paper is organized as follows. The system model is outlined in Section II. In Section III, a systematic design of STBC is proposed and a few code design examples are also given. The full diversity is proved under PIC group decoding in Section IV. In Section V, simulation results are presented. Finally, we conclude the paper in Section VI.

The following notations are used throughout this paper. Column vectors (matrices) are denoted by boldface lower (upper) case letters. Superscripts $(\cdot)^*$, $(\cdot)^t$ and $(\cdot)^H$ stand for conjugate, transpose, and conjugate transpose, respectively. $\mathbb{C}$ denotes the field of complex numbers  and $\mathbb{R}$ denotes the real field. $\mathbf{I}_n$ denotes the $n\times n$ identity matrix, and  $\mathbf{0}_{m\times n}$ denotes the $m\times n$ matrix  whose elements are all $0$. Additionally, $\{\cdot\}_R$ and $\{\cdot\}_I$ represent the real part and the imaginary part of   variables, respectively.

\section{System Model}\label{model}
Consider a MIMO system with $M$ transmit and $N$ receive antennas. Data symbols are first encoded into a space-time block code $\mathbf{X(s)}$ of size $T \times M$  where $T$ is block length of the codeword. In this paper, $\mathbf{X(s)}$ can be represented in a general dispersion form \cite{has} as follows:
\begin{eqnarray}
\mathcal{X}=\{\mathbf{X(s)}=\sum_{l=1}^L s_l A_l+s^{*}_lB_l\}
\end{eqnarray}
where the data symbols $\{{s_l}\},~l=1,2,\ldots,L$ are selected from a normalized complex constellation $\mathcal{A}$ such as QAM, and $A_l, B_l \in \mathbb{C}^{T\times M}$ are constant matrices called dispersion matrices. Then, the received signal from $N$ receive antennas can be arranged in a matrix $\mathbf{Y}\in \mathbb{C}^{T\times N}$ as follows
 \begin{eqnarray}\label{ry}
 \mathbf{Y}=\sqrt{\frac{\rho}{\mu}}\mathbf{X}(\mathbf{s})\mathbf{H}+\mathbf{W},
 \end{eqnarray}
where  $\mathbf{H}$ is the channel matrix of size $M\times N$ with the entries being independent and identically distributed (i.i.d) $\mathcal{CN}(0,1)$. The channels are assumed to experience the quasi-static fading. $\mathbf{W}\in \mathbb{C}^{T\times N}$ is the noise matrix whose elements are also i.i.d distributed $\mathcal{CN}(0,1)$. $\rho$ denotes the average signal-to-noise ratio (SNR) per receive antenna, and the transmitted power is normalized by the factor $\mu$ such that the average energy of the coded symbols transmitting from all antennas during one symbol period is one. We suppose that channel state information is available at receiver only.

To decode the transmitted sequence $\mathbf{s}$, we need to extract $\mathbf{s}$ from $\mathbf{X}(\mathbf{s})$. Through some operations, we can get an equivalent signal model from (\ref{ry}) as \cite{Shang}\cite{guoxia}
\begin{eqnarray}\label{eqyc}
\mathbf{y}=\sqrt{\frac{\rho}{\mu}}\mathcal{H}_c\mathbf{s}+\mathbf{w},
 \end{eqnarray}
 where $\mathbf{y} \in \mathbb{C}^{TN}\times 1$ is a received signal vector, $\mathbf{w}\in \mathbb{C}^{TN\times 1}$ is a noise vector, and $\mathcal{H}_c \in \mathbb{C}^{TN\times L}$ is an equivalent channel matrix. Denote  $\mathbf{y}=\mathbf{y}_R+j\mathbf{y}_I$,  $\mathbf{s}=\mathbf{s}_R+j\mathbf{s}_I$, and
  $\mathbf{w}=\mathbf{w}_R+j\mathbf{w}_I$. Then, we can rewrite (\ref{eqyc}) as a real matrix form given by
 \begin{eqnarray}\label{eqyr}
\left[\begin{array}{c}
\mathbf{y}_R\\
\mathbf{y}_I
 \end{array}\right]
=\sqrt{\frac{\rho}{\mu}}\mathcal{H}
\left[\begin{array}{c}
\mathbf{s}_R\\
\mathbf{s}_I
 \end{array}\right]
+
\left[\begin{array}{c}
\mathbf{w}_R\\
\mathbf{w}_I
 \end{array}\right],
 \end{eqnarray}
where $\mathcal{H} \in \mathbb{R}^{2TN\times 2L}$ has $2L$ real column vectors $\{\mathbf{g}_{l}\}$ for $l=1,2,\cdots,2L$.

In \cite{guoxia}, a new decoding scheme was proposed, referred to as PIC group decoding which aims to address the rate and complexity tradeoff of the code while achieving full diversity.  In the PIC group decoding, the equivalent channel matrix $\mathcal{H}_c \in \mathbb{C}^{TN\times L}$ is divided into a number of column groups $\{\mathbf{G}_1, \mathbf{G}_2,\cdots, \mathbf{G}_P\}$ with $l_p$ columns for group $\mathbf{G}_p$, $p=1,2,\cdots,P$, and $\sum_{p=1}^P l_p=L$. Then,  for group $\mathbf{G}_p$ a group ZF is applied   to cancel the interferences coming from all the other groups, i.e., $\{\mathbf{G}_1,\cdots, \mathbf{G}_{p-1}, \mathbf{G}_{p+1}, \cdots, \mathbf{G}_P\}$, followed by a joint decoding of symbols corresponding to the   group $\mathbf{G}_p$.  Note that   the  interference cancellation  (i.e., the group ZF) mainly involves with linear matrix computations, whose computational complexity  is small compared to the joint decoding with an exhaustive search of all candidate symbols in one group. To evaluate the decoding complexity of the PIC group decoding, we mainly focus on the computational complexity of the joint decoding of each group under the PIC group decoding algorithm. The joint decoding complexity can be characterized by the number of Frobenius norms calculated in the decoding process. In the PIC group decoding algorithm, the complexity is then   $\mathcal{O}= \sum_{p=1}^P |\mathcal{A}|^{l_p}$. It can be seen that the PIC group decoding provides a flexible decoding complexity which can vary from the ZF decoding complexity $L|\mathcal{A}|$ to the ML decoding complexity $|\mathcal{A}|^{L}$.

An SIC-aided PIC group decoding algorithm, namely PIC-SIC group decoding was also proposed in \cite{guoxia}. Similar to the  BLAST detection algorithm \cite{Foschini}, the PIC-SIC group decoding is performed  after removing the already-decoded symbol set from the received signals to reduce the interference. If each group has only one symbol, then the PIC-SIC group decoding will be equivalent to the BLAST detection.

In \cite{guoxia},  full-diversity STBC design criteria were derived
when the PIC group decoding and the PIC-SIC group decoding are used at the receiver. In the following, we cite the main results of the STBC design criteria proposed in \cite{guoxia}.

\begin{proposition}\label{prop1}
\cite[Theorem 1]{guoxia} [\emph{Full-Diversity Criterion under PIC Group Decoding}]

For an  STBC  $\mathbf{X}$  with the PIC group decoding, the full diversity is achieved when
  \begin{enumerate}
    \item   the code $\mathbf{X}$ satisfies the full rank criterion, i.e.,
it achieves full diversity when
the ML receiver is used; \emph{and}
    \item
    for any $p$, $1\leq p\leq P$, any non-zero linear combination over $\Delta \mathcal{A}$ of the vectors in the $p$th group $\mathbf{G}_p$ does not belong to the space linearly spanned by all the vectors in the remaining vector groups, as long as $\mathbf{H}\neq 0$, i.e.,
    \begin{eqnarray}
\sum_{\forall i\in \mathcal{I}_p} a_i \mathbf{g}_{i} \neq \sum_{\forall j\notin \mathcal{I}_p} c_j \mathbf{g}_{j}, \,\, \,\,\,\,a_i\in \Delta \mathcal{A},\mathrm{not\, all\, zero}, \,\,c_j\in \mathbb{C}
\end{eqnarray}
where $
\mathcal{I}_p=\{I_{p,1}, I_{p,2},\cdots, I_{p,l_p}\}$ is the index  set corresponding to the vector group $\mathbf{G}_{p}$ and $\Delta \mathcal{A}=\{S-\hat{S}, | S, \hat{S}\in \mathcal{A}\}$.

  \end{enumerate}
\end{proposition}

\begin{proposition}\label{prop2}
\cite{guoxia} [\emph{Full-Diversity Criterion under PIC-SIC Group Decoding}]

For an  STBC  $\mathbf{X}$  with the PIC-SIC group decoding, the full diversity is achieved when
  \begin{enumerate}
    \item   the code $\mathbf{X}$ satisfies the full rank criterion, i.e.,
it achieves full diversity when
the ML receiver is used; \emph{and}
    \item   at each decoding stage, for $\mathbf{G}_{q_1}$, which corresponds to the current to-be decoded symbol group $\mathbf{s}_{q_1}$,   any non-zero linear combination over $\Delta \mathcal{A}$ of the vectors in   $\mathbf{G}_{q_1}$ does not belong to the space linearly spanned by all the vectors in the remaining groups $\mathbf{G}_{q_2},  \cdots, \mathbf{G}_{q_L}$  corresponding to yet uncoded symbol groups, as long as $\mathbf{H}\neq 0$.
  \end{enumerate}
\end{proposition}

\section{Proposed STBC}
In this section, a systematic design of linear dispersion STBC is presented  and  two design examples are given for four and six transmit antennas, respectively.

\subsection{A Systematic Design}
Suppose that $M$ is even. Our proposed STBC $\mathbf{\Phi}$ is of size $T \times M$, and given by
\begin{eqnarray} \label{new}
\mathbf{\Phi}_{M,T,P}=\mathbf{A}_{M,T,P}+j~\mathbf{B}_{M,T,P}
\end{eqnarray}

where the codeword matrices for $\mathbf{A}_{M,T,P}$ and $\mathbf{B}_{M,T,P}$ are given by
\begin{eqnarray}\label{AB}
 \mathbf{A}_{M,T,P}
 =\left[\begin{array}{cc}
\mathbf{C}^1_R & \mathbf{C}^2_R \\
 -\mathbf{C}^2_R & \mathbf{C}^1_R
 \end{array}\right],
 ~\mathbf{B}_{M,T,P}\label{B}
  =\left[\begin{array}{cc}
\mathbf{C}^1_I & \mathbf{C}^2_I \\
 \mathbf{C}^2_I & -\mathbf{C}^1_I
 \end{array}\right].
\end{eqnarray}
Note that   $\mathbf{A}_{M,T,P}$ and $\mathbf{B}_{M,T,P}$ are both real matrices of size $T \times M$
($T=M+2P-2$).  $\mathbf{C}^i_R $ and $\mathbf{C}^i_I$ are real and imaginary parts of $\mathbf{C}^i\in \mathbb{C}^{\frac{T}{2}\times \frac{M}{2}}$ $(i=1,2)$ which is given by
\begin{eqnarray}\label{C}
 \mathbf{C}^i&=&\mathbf{C}^i_R+j~\mathbf{C}^i_I\\&=&
 \nonumber\left[\begin{array}{cccc}
                      {x}_{q_i+1}              &  0                    & \ldots        & 0 \\
                      {x}_{q_i+\frac{M}{2}+1}  &  {x}_{q_i+2}       & \ddots        &  \vdots \\
                       \vdots                     &  {x}_{q_i+\frac{M}{2}+2}           & \ddots & 0 \\
                      {x}_{q_i+(P-1)\frac{M}{2}+1}     &  \vdots            & \ddots        & {x}_{q_i+\frac{M}{2}} \\
                       0                          & {x}_{q_i+(P-1)\frac{M}{2}+2}              & \ddots        &  {x}_{q_i+M}\\
                       \vdots                          & \vdots               & \ddots        &  \vdots\\
                       0                          & \vdots              & \ddots        &  {x}_{q_i+P\frac{M}{2}}
                    \end{array}
\right],
\end{eqnarray}
with $q_i=(i-1)P\frac{M}{2}$ and the $p$ th diagonal layer from left to right written as the $\frac{M}{2} \times 1$ vector $\mathbf{X}^i_p~(i=1,2;~p=1,2,~\cdots,P)$, given by
\begin{eqnarray}\label{xir}
\mathbf{X}^i_p
=\left[\begin{array}{cccc}
{x}_{q_i+(p-1)\frac{M}{2}+1}
 & {x}_{q_i+(p-1)\frac{M}{2}+2}
&\cdots & {x}_{q_i+(p-1)\frac{M}{2}+ \frac{M}{2}}
\end{array}\right]^t.
\end{eqnarray}
Moreover,  the real and imaginary parts of $\mathbf{X}^i_p=\mathbf{X}^i_{p,R}+j~\mathbf{X}^i_{p,I}$ are given by respectively
\begin{eqnarray}
\mathbf{X}^i_{p,R} &=& \mathbf{\Theta}_{A} ~\mathbf{s}^i_{p,R} \label{roxir}\\
\mathbf{X}^i_{p,I} &=& \mathbf{\Theta}_{B} ~\mathbf{s}^i_{p,I} \label{roxii}
\end{eqnarray}
where   $\mathbf{\Theta}_{A}, \mathbf{\Theta}_{B}  \in \mathbb{R}^{\frac{M}{2} \times \frac{M}{2}}$ can be different linear transform matrices chosen from \cite{viterbo} and \cite{DRT}, and the $\frac{M}{2} \times 1$ vector $\mathbf{s}^i_p =\mathbf{s}^i_{p,R}+j~\mathbf{s}^i_{p,I}~(i=1,2;~p=1,2,~\ldots,P)$ are given by
\begin{eqnarray}\label{eqn:sip}
\mathbf{s}^i_p=\left[\begin{array}{cccc}
{s}_{q_i+(p-1)\frac{M}{2}+1}
 & {s}_{q_i+(p-1)\frac{M}{2}+2}
&\cdots & {s}_{q_i+(p-1)\frac{M}{2}+ \frac{M}{2}}
\end{array}\right]^t.
\end{eqnarray}

\begin{remark}
The  rate of the code $\mathbf{\Phi}_{M,T,P}$ is
\begin{eqnarray}\label{rate}
R=\frac{L}{T}=\frac{MP}{M+2P-2}
\end{eqnarray}
which is the same as that of STBCs with PIC group decoding proposed in \cite{wei} and \cite{jisit}.
\end{remark}

\begin{remark}
It should be mentioned that the code structure (\ref{new}) is similar to the one in \cite{jisit}. The main difference is that in \cite{jisit} a linear transform matrix is used for input complex symbol vectors in the code construction and the matrix does not have to be real-valued, whereas in the design of $\mathbf{\Phi}_{M,T,P}$ in (\ref{new}), two real linear transform matrices $\mathbf{\Theta}_A$ and
$\mathbf{\Theta}_B$ are used.
 Later, we will see that the proposed code in (\ref{new}) with PIC group decoding of real-valued signals yields lower decoding complexity than the one in \cite{jisit}.
\end{remark}

%
%
%

\subsection{Code Design Examples for $P=2$}
1) For Four Transmit Antennas $M=4$

Consider the case with $M=4$ transmit antennas. According to the design in (\ref{new}), we have
\begin{equation}\label{m4}
\mathbf{\Phi}_{4,6,2}=\mathbf{A}_{4,6,2}+j~\mathbf{B}_{4,6,2}
\end{equation}
where
\begin{eqnarray}
  \mathbf{A}_{4,6,2} = \left[\begin{array}{cccc}
                        {x}_{1,R}& 0 &  {x}_{5,R} & 0 \\
                       {x}_{3,R} & {x}_{2,R} & {x}_{7,R} & {x}_{6,R}\\
                       0  &  {x}_{4,R} & 0  &  {x}_{8,R}\\
                        -{x}_{5,R} & 0 & {x}_{1,R}& 0\\
                        -{x}_{7,R}& -{x}_{6,R} & {x}_{3,R} & {x}_{2,R}\\
                        0  &  -{x}_{8,R} & 0  &  {x}_{4,R}
                     \end{array}
\right],
 \end{eqnarray}
 with
 $\left[
 \begin{array}{cc}
    {x}_{\{2(i-1)+1\},R} &   {x}_{\{2(i-1)+2\},R}
 \end{array}\right]^t
=\mathbf{\Theta}_{A}
  \left[\begin{array}{cc}
s_{\{2(i-1)+1\},R} & s_{\{2(i-1)+2\},R}
\end{array}\right]^t$  for $i=1,2,3,4$,
and
\begin{eqnarray}
   \mathbf{B}_{4,6,2} =\left[\begin{array}{cccc}
                       {x}_{1,I} & 0 &  {x}_{5,I} & 0 \\
                       {x}_{3,I} & {x}_{2,I} & {x}_{7,I} & {x}_{6,I}\\
                       0  &  {x}_{4,I} & 0  &  {x}_{8,I}\\
                        {x}_{5,I} & 0 & -{x}_{1,I} & 0\\
                        {x}_{7,I} & {x}_{6,I} & -{x}_{3,I} & -{x}_{2,I}\\
                        0  &  {x}_{8,I} & 0  &  -{x}_{4,I}
                     \end{array}
\right]
 \end{eqnarray}
with
 $\left[
 \begin{array}{cc}
    {x}_{\{2(i-1)+1\},I} &   {x}_{\{2(i-1)+2\},I}
 \end{array}\right]^t
=\mathbf{\Theta}_{B}
  \left[\begin{array}{cc}
s_{\{2(i-1)+1\},I} & s_{\{2(i-1)+2\},I}
\end{array}\right]^t$  for $i=1,2,3,4$.

For simplicity, the same linear transform matrix $\mathbf{\Theta}_{2\times2}$ is used for $\mathbf{\Theta}_{A}$ and $\mathbf{\Theta}_{B}$   as
\begin{eqnarray}\label{rota4}
\mathbf{\Theta}_{2\times 2}=\left[\begin{array}{cc}
\cos\alpha & \sin\alpha\\
-\sin\alpha & \cos\alpha
\end{array}\right],
\end{eqnarray}
with $\alpha=1.02$ \cite{guoxia}.

%

Then, the codeword matrix of $\mathbf{\Phi}_{4,6,2}$ is written as
\begin{eqnarray}\label{M4}
\mathbf{\Phi}_{4,6,2}
  =\left[\begin{array}{cccc}
                        {x}_{1,R}+j{x}_{1,I} & 0 &  {x}_{5,R}+j{x}_{5,I} & 0 \\
                       {x}_{3,R}+j{x}_{3,I} & {x}_{2,R}+j{x}_{2,I} & {x}_{7,R}+j{x}_{7,I} & {x}_{6,R}+j{x}_{6,I}\\
                       0  &  {x}_{4,R}+j{x}_{4,I} & 0  &  {x}_{8,R}+j{x}_{8,I}\\
                        -{x}_{5,R}+j{x}_{5,I} & 0 & {x}_{1,R}-j{x}_{1,I} & 0\\
                        -{x}_{7,R}+j{x}_{7,I} & -{x}_{6,R}+j{x}_{6,I} & {x}_{3,R}-j{x}_{3,I} & {x}_{2,R}-j{x}_{2,I}\\
                        0  &  -{x}_{8,R}+j{x}_{8,I} & 0  &  {x}_{4,R}-j{x}_{4,I}
                     \end{array}
\right],
  \end{eqnarray}


The rate of the code  $\mathbf{\Phi}_{4,6,2}$ is $4/3$ and equal to that of  $\mathbf{C}_{4,6,2}$ in \cite[Eq. (29)]{wei} and $\mathbf{B}_{4,6,2}$ in \cite[Eq. (37)]{jisit}.

2) For Six Transmit Antennas $M=6$

For given $T=8$, the code $\mathbf{\Phi}_{6,8,2}$ with six transmit antennas is designed as follows
\begin{equation}\label{m6}
\mathbf{\Phi}_{6,8,2}=\mathbf{A}_{6,8,2}+j~\mathbf{B}_{6,8,2}
\end{equation}
where
\begin{eqnarray}
  \mathbf{A}_{6,8,2} = \left[\begin{array}{cccccc}
                        {x}_{1,R}& 0 &0 &  {x}_{7,R} & 0 & 0 \\
                       {x}_{4,R} & {x}_{2,R} &0& {x}_{10,R} & {x}_{8,R}&0\\
                       0  & {x}_{5,R} & {x}_{3,R}& 0  & {x}_{11,R} & {x}_{9,R}\\
                       0 & 0 &{x}_{6,R}& 0 & 0 &{x}_{12,R}\\
                     -{x}_{7,R} & 0 & 0 & {x}_{1,R}& 0 &0\\
                      -{x}_{10,R} & -{x}_{8,R}& 0 & {x}_{4,R} & {x}_{2,R} &0\\
                        0  & -{x}_{11,R} & -{x}_{9,R} & 0  & {x}_{5,R} & {x}_{3,R}\\
                        0 & 0 &-{x}_{12,R}& 0 & 0 &{x}_{6,R}
                     \end{array}
\right]
 \end{eqnarray}
with
\begin{eqnarray}
\left[\begin{array}{c}\label{rotar6}
{x}_{\{3(i-1)+1\},R} \\ {x}_{\{3(i-1)+2\},R}\\ {x}_{\{3(i-1)+3\},R}
\end{array}\right]
=\mathbf{\Theta}_A
  \left[\begin{array}{c}
s_{\{3(i-1)+1\},R}\\ s_{\{3(i-1)+2\},R}\\ s_{\{3(i-1)+3\},R}
\end{array}\right],
\end{eqnarray}
and
\begin{eqnarray}
   \mathbf{B}_{6,8,2} =\left[\begin{array}{cccccc}
                       {x}_{1,I}& 0 &0 &  {x}_{7,I} & 0 & 0 \\
                       {x}_{4,I} & {x}_{2,I} &0& {x}_{10,I} & {x}_{8,I}&0\\
                       0  & {x}_{5,I} & {x}_{3,I}& 0  & {x}_{11,I} & {x}_{9,I}\\
                       0 & 0 &{x}_{6,I}& 0 & 0 &{x}_{12,I}\\
                     {x}_{7,I} & 0 & 0 & -{x}_{1,R}& 0 &0\\
                      {x}_{10,I} & {x}_{8,I}& 0 & -{x}_{4,I} & -{x}_{2,I} &0\\
                        0  & {x}_{11,I} & {x}_{9,I} & 0  & -{x}_{5,I} & -{x}_{3,I}\\
                        0 & 0 &{x}_{12,I}& 0 & 0 &-{x}_{6,I}
                     \end{array}
\right]
 \end{eqnarray}
with
\begin{eqnarray}
\left[\begin{array}{c}\label{rotai6}
{x}_{\{3(i-1)+1\},I} \\ {x}_{\{3(i-1)+2\},I}\\ {x}_{\{3(i-1)+3\},I}
\end{array}\right]
=\mathbf{\Theta}_B
  \left[\begin{array}{c}
s_{\{3(i-1)+1\},I}\\ s_{\{3(i-1)+2\},I}\\ s_{\{3(i-1)+3\},I}
\end{array}\right],
\end{eqnarray}
for $i=1,2,3,4$.

The same linear transform matrix $\mathbf{\Theta}_{3\times3}$ is used for $\mathbf{\Theta}_{A}$ and $\mathbf{\Theta}_{B}$   as \cite{DRT}
\begin{eqnarray}
\mathbf{\Theta}_{3\times3}=\left[\begin{array}{ccc}
0.745 & -0.582 & -0.326\\
-0.326 & 0.745 &-0.582\\
0.582& 0.326 & 0.745
\end{array}\right]
\end{eqnarray}

Then, the codeword $\mathbf{\Phi}_{6,8,2}$ can be written as
\begin{equation}
\begin{split}
   &\mathbf{\Phi}_{6,8,2}=\\
    &\left[\begin{array}{cccccc}
                       {x}_{1,R}+j{x}_{1,I}& 0 &0 &  {x}_{7,R}+j{x}_{7,I} & 0 & 0 \\
                       {x}_{4,R}+j{x}_{4,I} & {x}_{2,R}+j{x}_{2,I} &0& {x}_{10,R}+j{x}_{10,I} & {x}_{8,R}+j{x}_{8,I}&0\\
                       0  & {x}_{5,R}+j{x}_{5,I} & {x}_{3,R}+j{x}_{3,I}& 0  & {x}_{11,R}+j{x}_{11,I} & {x}_{9,R}+j{x}_{9,I}\\
                       0 & 0 &{x}_{6,R}+j{x}_{6,I}& 0 & 0 & {x}_{12,R}+j{x}_{12,I}\\
                     -{x}_{7,R}+j{x}_{7,I} & 0 & 0 & {x}_{1,R}-j{x}_{1,I}& 0 &0\\
                      -{x}_{10,R}+j{x}_{10,I} & -{x}_{8,R}+{x}_{8,I}& 0 & {x}_{4,R}-{x}_{4,I} & {x}_{2,R}-j{x}_{2,I} &0\\
                        0  & -{x}_{11,R}+j{x}_{11,I} & -{x}_{9,R}+j{x}_{9,I} & 0  & {x}_{5,R}-j{x}_{5,I} & {x}_{3,R}-j{x}_{3,I}\\
                        0 & 0 &-{x}_{12,R}+j{x}_{12,I}& 0 & 0 &{x}_{6,R}-j{x}_{6,I}
                     \end{array}
\right].
\end{split}
 \end{equation}

The code rate for $\mathbf{\Phi}_{6,8,2}$ is $3/2$.

\section{Full Diversity of Proposed STBC with PIC Group Decoding}
In this section, we prove that our proposed STBC can obtain full diversity under PIC group decoding  and have a lower   decoding complexity compared with \cite{wei} and \cite{jisit}.

\subsection{Achieving Full Diversity with ML Decoding}

Define $\check{s}=s-\hat{s}$ as the difference between symbols $s$ and $\hat{s}$. Following the proof of \cite[Theorem 1]{jisit}, three cases should be considered separately in terms of $\check{\mathbf{s}}_R$ and $\check{\mathbf{s}}_I$ as follows

\begin{enumerate}
\item Both $\check{\mathbf{s}}_R \neq \mathbf{0}$ and $\check{\mathbf{s}}_I \neq \mathbf{0}$  \label{case1}

Consider $\check{\mathbf{s}}_R \neq \mathbf{0}$.
After some row/column permutations, a different codeword matrix $\check{\mathbf{A}}_{M,T,P}=\mathbf{A}_{M,T,P}-\hat{\mathbf{A}}_{M,T,P}$ can be written  as follows
\begin{equation}\label{ex}
  \check{\mathbf{A}}_{M,T,P}= \left[\begin{array}{cccc}
                     \check{\mathbf{T}}_1 & \mathbf{0}  & \cdots  &  \mathbf{0} \\
                     \check{\mathbf{T}}_{\frac{M}{2}+1} & \check{\mathbf{T}}_2 & \ddots  & \vdots  \\
                      \vdots & \check{\mathbf{T}}_{\frac{M}{2}+2} &  \ddots &  \mathbf{0} \\
                     \check{\mathbf{T}}_{(P-1)\frac{M}{2}+1}  & \vdots  & \ddots & \check{\mathbf{T}}_{M/2}  \\
                     \mathbf{0} & \check{\mathbf{T}}_{(P-1)\frac{M}{2}+2} & \ddots & \check{\mathbf{T}}_M\\
                     \vdots & \mathbf{0} & \ddots& \vdots\\
                     \mathbf{0} & \vdots & \ddots& \check{\mathbf{T}}_{P\frac{M}{2}}
                  \end{array}
                 \right]
 \end{equation}
 where $\mathbf{0}$ is a $2 \times 2$ matrix and
 \begin{eqnarray}
\check{\mathbf{T}}_{i} = \left[\begin{array}{cc}
    \check{x}_{i,R} & \check{x}_{\{i+\frac{MP}{2}\},R} \\
    -\check{x}_{\{i+\frac{MP}{2}\},R} & \check{x}_{i,R}\\
  \end{array}
  \right],i=1,2.
 \end{eqnarray}

 From (\ref{roxir}) and (\ref{roxii}), we deduce that there exists at least one vector $\mathbf{X}^i_p$ such that $\mathbf{X}^i_p-\hat{\mathbf{X}}^i_p\neq \mathbf{0}$, $p=1,2,\ldots,P$, because the signal space diversity is obtained from the linear transform matrix $\mathbf{\Theta}_{A}$.  Then, we have that $\check{\mathbf{A}}_{M,T,P}=\mathbf{A}_{M,T,P}-\hat{\mathbf{A}}_{M,T,P}$ is full rank, which can be proved with a  similar proof given in \cite[Theorem 1]{jisit}. Hence, $\mathbf{A}_{M,T,P}$ can guarantee full diversity with ML decoding. Likewise, it is obvious that $\mathbf{B}_{M,T,P}$ can also achieve full diversity since $\check{\mathbf{s}}_I \neq \mathbf{0}$.

Therefore, the code $\mathbf{\Phi}_{M,T,P}$ can achieve full diversity under ML decoding.

\item $\check{\mathbf{s}}_R \neq \mathbf{0}$ only \label{case2}

As we mentioned in case \ref{case1}), $\mathbf{A}_{M,T,P}$ can achieve full diversity under ML decoding if $\check{\mathbf{s}}_R \neq \mathbf{0}$. Considering $\mathbf{A}_{M,T,P}$ forms the real part in (\ref{new}), the code $\mathbf{\Phi}_{M,T,P}$ can achieve full diversity under ML decoding.

\item  $\check{\mathbf{s}}_I \neq \mathbf{0} $ only

Similar to case \ref{case2}), $\mathbf{B}_{M,T,P}$ being the imaginary part of our proposed code can achieve full diversity under ML decoding, which is sufficient to prove that $\mathbf{\Phi}_{M,T,P}$ has a property of full diversity.
\end{enumerate}

By observing all three cases, we conclude that the proposed code in (\ref{new}) can achieve full diversity under ML decoding.

\subsection{Achieving Full Diversity with PIC Group Group Decoding when $P=2$}

Compared with the PIC grouping schemes derived in \cite{wei} and \cite{jisit}, the separated linear transform of real and imaginary parts of the information symbols in the proposed code contributes to the real symbol decoding. In the following, we show the main result of the proposed STBC when a PIC group decoding with a particular grouping scheme is used at the receiver, as follows.

\begin{theorem}\label{tpic}
Consider a MIMO system with $M$ transmit antennas and $N$ receive antennas over block fading channels. The  STBC  as describe in (\ref{new}) with two diagonal layers in each submatrix is used at the transmitter. The real equivalent channel matrix is $\mathcal{H}\in \mathbb{R}^{2TN\times 2L}$. If the received signal is decoded using the PIC group decoding with the grouping scheme $I={\{I_1,I_2,\ldots,I_8\}}$, where $I_p=\{(p-1)M/2+1,\ldots,pM/2\}$ for $p=1,2,\ldots,8$, i.e., the size of each real group is equal to $M/2$, then the code $\mathbf{\Phi}_{M,T,2}$ achieves the full diversity.
\end{theorem}

\begin{corollary}\label{m4co}
For the proposed code with $M=4$ transmit antennas in (\ref{m4}),   real symbol pairwise ML decoding is achieved in each group, which is equivalent to single complex symbol  ML decoding.
\end{corollary}

Table I shows the comparison of PIC group decoding complexity between the new code in (\ref{m4}) and the codes in \cite{wei} and \cite{jisit}. According to this table, it is obvious that the proposed code for $M=4$  transmit antennas further reduce the decoding complexity to real symbol pairwise (i.e., single complex symbol) decoding in each PIC group.

In order to prove \emph{Theorem \ref{tpic}}, let us first introduce the following definition and lemma.

\begin{definition}\label{orde}
Let $\mathcal{V}_0, \mathcal{V}_1, \ldots, \mathcal{V}_{n}$ be $n$ groups of vectors. Vector groups $\mathcal{V}_0, \mathcal{V}_1, \ldots, \mathcal{V}_{n}$ are said to be orthogonal if for $0 \leq k \leq n$, $\mathcal{V}_{k}$ is orthogonal to the remaining vector groups $\mathcal{V}_0, \mathcal{V}_1, \ldots, \mathcal{V}_{k-1},\mathcal{V}_{k+1},\ldots, \mathcal{V}_{n}$.
\end{definition}

\begin{lemma}\label{leq}
Consider the system  described in \emph{Theorem \ref{tpic}} with $N=1$ as follows
\begin{eqnarray}\label{eqn:ye3}
\left[\begin{array}{c}
\mathbf{y}_{R}^1\\
\mathbf{y}_{R}^2\\
\mathbf{y}_{I}^1\\
\mathbf{y}_{I}^2   \end{array}
  \right]=
 \sqrt{\frac{\rho}{\mu}}\mathcal{H}
    \left[\begin{array}{c}
  \mathbf{s}^1_{1,R}\\ \mathbf{s}^1_{2,R}\\ \mathbf{s}^2_{1,R} \\ \mathbf{s}^2_{2,R}\\ \mathbf{s}^1_{1,I}\\ \mathbf{s}^1_{2,I}\\ \mathbf{s}^2_{1,I}\\ \mathbf{s}^2_{2,I}
    \end{array}
  \right]
  +\left[\begin{array}{c}
\mathbf{w}_{R}^1\\
\mathbf{w}_{R}^2\\
\mathbf{w}_{I}^1\\
\mathbf{w}_{I}^2
\end{array}\right],
\end{eqnarray}
where the $\frac{M}{2} \times 1$ vector $\mathbf{s}^i_p =\mathbf{s}^i_{p,R}+j~\mathbf{s}^i_{p,I}$  are given by (\ref{eqn:sip}) for $i=1,2$ and $p=1,2$.
The   equivalent channel matrix $\mathcal{H} \in \mathbb{R}^{2T\times 2L}$  is expressed as
\begin{eqnarray}
\mathcal{H}&=& \left[\begin{array}{cccccccc}\nonumber\label{lemmaeq}
\mathcal{H}^1_{1,R}\mathbf{\Theta}_{A}&\mathcal{H}^1_{2,R}\mathbf{\Theta}_{A}&\mathcal{H}^2_{1,R}\mathbf{\Theta}_{A}&\mathcal{H}^2_{2,R}\mathbf{\Theta}_{A}&-\mathcal{H}^1_{1,I}\mathbf{\Theta}_{B}&-\mathcal{H}^1_{2,I}\mathbf{\Theta}_{B}&-\mathcal{H}^2_{1,I}\mathbf{\Theta}_{B}&-\mathcal{H}^2_{2,I}\mathbf{\Theta}_{B}\\
\mathcal{H}^2_{1,R}\mathbf{\Theta}_{A}&\mathcal{H}^2_{2,R}\mathbf{\Theta}_{A}&-\mathcal{H}^1_{1,R}\mathbf{\Theta}_{A}&-\mathcal{H}^1_{2,R}\mathbf{\Theta}_{A}&\mathcal{H}^2_{1,I}\mathbf{\Theta}_{B}&\mathcal{H}^2_{2,I}\mathbf{\Theta}_{B}&-\mathcal{H}^1_{1,I}\mathbf{\Theta}_{B}&-\mathcal{H}^1_{2,I}\mathbf{\Theta}_{B}\\
   \mathcal{H}^1_{1,I}\mathbf{\Theta}_{A}&\mathcal{H}^1_{2,I}\mathbf{\Theta}_{A}&\mathcal{H}^2_{1,I}\mathbf{\Theta}_{A}&\mathcal{H}^2_{2,I}\mathbf{\Theta}_{A}&\mathcal{H}^1_{1,R}\mathbf{\Theta}_{B}&\mathcal{H}^1_{2,R}\mathbf{\Theta}_{B}&\mathcal{H}^2_{1,R}\mathbf{\Theta}_{B}&\mathcal{H}^2_{2,R}\mathbf{\Theta}_{B}\\
  \mathcal{H}^2_{1,I}\mathbf{\Theta}_{A}&\mathcal{H}^2_{2,I}\mathbf{\Theta}_{A}&-\mathcal{H}^1_{1,I}\mathbf{\Theta}_{A}&-\mathcal{H}^1_{2,I}\mathbf{\Theta}_{A}&-\mathcal{H}^2_{1,R}\mathbf{\Theta}_{B}&-\mathcal{H}^2_{2,R}\mathbf{\Theta}_{B}&\mathcal{H}^1_{1,R}\mathbf{\Theta}_{B}&\mathcal{H}^1_{2,R}\mathbf{\Theta}_{B}  \end{array}
\right]  \\
&=&\left[\begin{array}{cccccccc}\label{picl}
  \mathbf{G}_1 &\mathbf{G}_2 &\ldots & \mathbf{G}_8
  \end{array}
\right],
\end{eqnarray}
where
\begin{eqnarray}
\mathcal{H}_{p,R}^i=\left[\begin{array}{c}
      \mathbf{0}_{(p-1)\times (M/2)}\\
      \mathrm{diag}(\mathbf{h}_R^{i})\\
      \mathbf{0}_{(2-p)\times (M/2)}
      \end{array}
  \right], {\rm{and}~}
\mathcal{H}_{p,I}^i=\left[\begin{array}{c}
      \mathbf{0}_{(p-1)\times (M/2)}\\
      \mathrm{diag}(\mathbf{h}_I^{i})\\
      \mathbf{0}_{(2-p)\times (M/2)}
      \end{array}
  \right],
\end{eqnarray}
for $p=1,2$ and $i=1,2$. The channel coefficient vector $\mathbf{h} \in \mathbb{C}^{M\times1}$ is evenly divided into two groups with $\mathbf{h}^1=[\begin{array}{cccc}
                       h_1 & h_2 & \cdots & h_{\frac{M}{2}}
                     \end{array}
 ]^t$ and $\mathbf{h}^2=[\begin{array}{cccc}
                       h_{\frac{M}{2}+1} & h_{\frac{M}{2}+2} & \cdots & h_{M}
                     \end{array}
 ]^t$.  $h_j$ is the channel gain from the $j$th transmit antenna to the single receive antenna for $j=1,2,\cdots, M$.

\end{lemma}

A proof of \emph{Lemma 1} is given in Appendix.

\begin{proof}[Proof of Theorem \ref{tpic}]


Note that \cite[Corollary 1]{guoxia} proves that the full diversity conditions only need to be proved for one receive antenna case. Thus, we only consider the MISO system model (i.e. $N=1$).

First, after some column/row permutations, (\ref{lemmaeq}) can be rewritten as
\begin{eqnarray}\label{eqhp1}
\mathcal{H}^{'} &=& \left[\begin{array}{cccccccc}
                     [\mathbf{G}_1^{'}  & \mathbf{G}_3^{'} ]& [\mathbf{G}_5^{'} & \mathbf{G}_7^{'}]& [\mathbf{G}_2^{'} & \mathbf{G}_4^{'}]& [\mathbf{G}_6^{'} & \mathbf{G}_8^{'}]
                   \end{array}
\right],\nonumber\\
&=&\left[\begin{array}{cccc}
                       \mathcal{F}^{R}_1 & -\mathcal{F}^{I}_1 & \mathbf{0}_{2\times M} & \mathbf{0}_{2\times M} \\
                        \mathcal{F}^{I}_1 & \mathcal{F}^{R}_1  & \mathbf{0}_{2\times M} & \mathbf{0}_{2\times M} \\
                       \mathcal{F}^{R}_2 & -\mathcal{F}^{I}_2  & \mathcal{F}^{R}_1 & -\mathcal{F}^{I}_1  \\
                        \mathcal{F}^{I}_2 & \mathcal{F}^{R}_2 & \mathcal{F}^{I}_1 & \mathcal{F}^{R}_1\\
                         \vdots            &   \vdots  &\mathcal{F}^{R}_2 & -\mathcal{F}^{I}_2\\
                        \vdots            &   \vdots         & \mathcal{F}^{I}_2 & \mathcal{F}^{R}_2 \\
                        \mathcal{F}^{R}_{\frac{M}{2}} & -\mathcal{F}^{I}_{\frac{M}{2}} &\vdots  &\vdots  \\
                        \mathcal{F}^{I}_{\frac{M}{2}} & \mathcal{F}^{R}_{\frac{M}{2}} &\vdots   &   \vdots   \\
                        \mathbf{0}_{2\times M} & \mathbf{0}_{2\times M} &\mathcal{F}^{R}_{\frac{M}{2}} & -\mathcal{F}^{I}_{\frac{M}{2}}\\
                        \mathbf{0}_{2\times M} & \mathbf{0}_{2\times M} & \mathcal{F}^{I}_{\frac{M}{2}} & \mathcal{F}^{R}_{\frac{M}{2}}
                      \end{array}
\right],
\end{eqnarray}
where  both $\mathcal{F}^{R}_j$ and $\mathcal{F}^{I}_j$ are $2\times M$ real matrix given by
\begin{eqnarray}
\mathcal{F}^R_j&=&\left[\begin{array}{cc}\label{FR}
                      \mathbf{f}_{j,j_R} & \mathbf{f}_{j,{\{j+\frac{M}{2}\}}_R} \\
                     \mathbf{f}_{j,{\{j+\frac{M}{2}\}}_R} &   -\mathbf{f}_{j,j_R}
                   \end{array}
\right],\\
\mathcal{F}^I_j&=&\left[\begin{array}{cc}\label{FI}
                      \mathbf{f}_{j,j_I} & \mathbf{f}_{j,{\{j+\frac{M}{2}\}}_I} \\
                     -\mathbf{f}_{j,{\{j+\frac{M}{2}\}}_I} &   \mathbf{f}_{j,j_I}
                   \end{array}
\right],
\end{eqnarray}
for $j=1,2,\cdots,\frac{M}{2}$. $\mathbf{f}_{i,j_{(.)}}= \Theta_{i} h_{j_{(.)}}$ is a $1 \times \frac{M}{2}$ real vector with $\Theta_i$ being the $i$th row of the linear transform matrix $\mathbf{\Theta}$ for $i=1,2,\ldots,\frac{M}{2}$ and with $h_{j_R}$ and $h_{j_I}$ being the real and imaginary part of $h_j$ for $j=1,2,\ldots,M$, respectively. It is worthwhile to mention that from  (\ref{lemmaeq}), $\mathcal{F}^R_j$ in  $[\mathbf{G}_1^{'}, ~ \mathbf{G}_3^{'} ]$ and $  [\mathbf{G}_2^{'},~  \mathbf{G}_4^{'}]$ are related to $\mathbf{\Theta}_A$, while  $\mathcal{F}^I_j$ in $[\mathbf{G}_5^{'},~ \mathbf{G}_7^{'}]$ and $[\mathbf{G}_6^{'},~  \mathbf{G}_8^{'}]$ are associated with $\mathbf{\Theta}_B$. In Appendix, it is shown that the orthogonality between each groups is irrelevant to the linear transform matrices  $\mathbf{\Theta}_A$ and $\mathbf{\Theta}_B$. Therefore, for simplicity  $\mathbf{\Theta}$ is used for both
$\mathbf{\Theta}_A$ and $\mathbf{\Theta}_B$.

Next, we   prove that any non-zero linear combination of the vectors in $\mathbf{G}_1^{'}$  over $\Delta \mathcal{A}$ does not belong to the space linearly spanned by all the vectors in the vector groups $\mathbf{G}_2^{'},~\mathbf{G}_3^{'},~\ldots,~\mathbf{G}_8^{'}$. for any $\mathbf{h}\neq 0$, i.e.,
\begin{eqnarray}\label{ture1}
\sum_{\forall  \mathbf{g}_{i}\subset \mathbf{G}_1^{'}} a_i \mathbf{g}_{i} \neq \sum_{\forall \mathbf{g}_{j}\subset \{\mathbf{G}_2^{'},~\mathbf{G}_3^{'},~\ldots,~\mathbf{G}_8^{'} \}}c_j \mathbf{g}_{j}, ~~~~a_i\in \Delta \mathcal{A},{~\rm{not\, all\, zero}}, ~c_j\in \mathbb{C}.
\end{eqnarray}
where $\mathbf{g}_{i}$ is a column vector.

%

For any  nonzero $\mathbf{h}=\mathbf{h}_R+j\mathbf{h}_I$, we have following three cases.
\begin{description}
  \item[A)]  if  $\mathbf{h}_R\neq \mathbf{0}$ and $\mathbf{h}_I\neq \mathbf{0}$, then
it must exist a minimum index $j$ ($1\leq j\leq M/2$) such that $\mathcal{F}^R_j$ is nonzero and a minimum index $l$ ($1\leq l\leq M/2$)  such that     $\mathcal{F}^I_l$  is nonzero. Therefore, $\mathcal{F}^{R}_1,\cdots, \mathcal{F}^R_{j-1}$  must be all zeros and  $\mathcal{F}^{I}_1,\cdots, \mathcal{F}^I_{l-1}$ must be all zeros, too.
  \item[B)]   if  $\mathbf{h}_R\neq \mathbf{0}$ and $\mathbf{h}_I= \mathbf{0}$, then
it must exist a minimum index $j$ ($1\leq j\leq M/2$) such that $\mathcal{F}^R_j$ is nonzero. Therefore, $\mathcal{F}^{R}_1,\cdots, \mathcal{F}^R_{j-1}$  must be all zeros and $\mathcal{F}^{I}_1,\cdots, \mathcal{F}^I_{M/2}$ must be all zeros, too.
  \item[C)]    $\mathbf{h}_I\neq \mathbf{0}$ and $\mathbf{h}_R= \mathbf{0}$, then
it must exist a minimum index $l$ ($1\leq l\leq M/2$) such that $\mathcal{F}^I_l$ is nonzero. Therefore, $\mathcal{F}^{I}_1,\cdots, \mathcal{F}^I_{l-1}$ must be all zeros and $\mathcal{F}^{R}_1,\cdots, \mathcal{F}^R_{M/2}$ must be all zeros, too.
\end{description}

Next, we first focus on the case of A).
The proof is presented in terms of $j$ and $l$.

\begin{description}
\item[A1)] $j=l$

In this case, (\ref{eqhp1}) can be expressed as
\begin{eqnarray}\label{eqhp2}
\mathcal{H}^{'}
&=&\left[\begin{array}{cccc}
                       \mathbf{0} & \mathbf{0} & \mathbf{0} & \mathbf{0}\\
                       \vdots & \vdots & \vdots & \vdots\\
                       \mathbf{0}& \mathbf{0} & \mathbf{0}& \mathbf{0}\\
                        \mathcal{F}^{R}_j &  -\mathcal{F}^{I}_j &\mathbf{0} &\mathbf{0}\\
                         \mathcal{F}^{I}_j &  \mathcal{F}^{R}_j & \mathbf{0} &\mathbf{0}\\
                             \vdots & \vdots &   \mathcal{F}^{R}_j &  -\mathcal{F}^{I}_j\\
                        \vdots & \vdots & \mathcal{F}^{I}_j &  \mathcal{F}^{R}_j \\
                         \vdots & \vdots & \vdots & \vdots\\
                        \mathcal{F}^{R}_{\frac{M}{2}} & -\mathcal{F}^{I}_{\frac{M}{2}} &\vdots &\vdots \\
                        \mathcal{F}^{I}_{\frac{M}{2}} & \mathcal{F}^{R}_{\frac{M}{2}} &\vdots &\vdots \\
                          \mathbf{0} & \mathbf{0} &\mathcal{F}^R_{\frac{M}{2}}& -\mathcal{F}^I_{\frac{M}{2}}\\
                           \mathbf{0} & \mathbf{0} &\mathcal{F}^I_{\frac{M}{2}}& \mathcal{F}^R_{\frac{M}{2}}\\
                      \end{array}
\right].
\end{eqnarray}
where $\mathbf{0}=\mathbf{0}_{2\times M}$.
%
%
%

By observing the $(4j-3)$th row to the $(4j)$th row in (\ref{eqhp2}), the vector groups $\mathbf{G}_2^{'},~\mathbf{G}_4^{'},~\mathbf{G}_6^{'},~\mathbf{G}_8^{'}$ are all zeros, and $\mathbf{G}_1^{'}$ is orthogonal to  the vector groups $\mathbf{G}_3^{'},~\mathbf{G}_5^{'},~\mathbf{G}_7^{'}$.  Thus, it is obvious that  in these four rows, any non-zero linear combination of the vectors in $\mathbf{G}_1^{'}$  over $\Delta \mathcal{A}$ does not belong to the space linearly spanned by all the vectors in the vector groups $\mathbf{G}_2^{'},~\mathbf{G}_4^{'},~\mathbf{G}_6^{'},~\mathbf{G}_8^{'}$.

Furthermore, according to \emph{Definition \ref{orde}}, the vector groups $\mathbf{G}_1^{'}, ~\mathbf{G}_3^{'}, ~\mathbf{G}_5^{'}, ~\mathbf{G}_7^{'}$ are orthogonal in these four rows. Consequently,  in these four rows, any non-zero linear combination of the vectors in $\mathbf{G}_1^{'}$  over $\Delta \mathcal{A}$ does not belong to the space linearly spanned by all the vectors in the vector groups $\mathbf{G}_2^{'},~\mathbf{G}_3^{'},~\ldots,~\mathbf{G}_8^{'}$. Considering all rows in (\ref{eqhp2}), any non-zero linear combination of the vectors in $\mathbf{G}_1^{'}$  over $\Delta \mathcal{A}$ does not belong to the space linearly spanned by all the vectors in the vector groups $\mathbf{G}_2^{'},~\mathbf{G}_3^{'},~\ldots,~\mathbf{G}_8^{'}$.


\item[A2)] $j>l$

In this case, (\ref{eqhp1}) can be expressed as
\begin{eqnarray}\label{eqhp3}
\mathcal{H}^{'}
&=&\left[\begin{array}{cccc}
                       \mathbf{0} & \mathbf{0} & \mathbf{0} & \mathbf{0}\\
                       \vdots & \vdots & \vdots & \vdots\\
                       \mathbf{0}& \mathbf{0} & \vdots& \vdots\\
                        \mathbf{0} &  -\mathcal{F}^{I}_l &\mathbf{0} &\mathbf{0}\\
                         \mathcal{F}^{I}_l &  \mathbf{0} &\mathbf{0} &\mathbf{0}\\
                       \vdots & \vdots & \mathbf{0}& -\mathcal{F}^I_l \\
                        \vdots& \vdots& \mathcal{F}^I_l& \mathbf{0} \\
                        \vdots & \vdots & \vdots & \vdots\\
                        \mathbf{0}& -\mathcal{F}^{I}_{j-1}& \vdots & \vdots\\
                        \mathcal{F}^{I}_{j-1} & \mathbf{0}& \vdots & \vdots\\
                        \mathcal{F}^{R}_j & -\mathcal{F}^{I}_j & \mathbf{0}& -\mathcal{F}^{I}_{j-1}\\
                        \mathcal{F}^{I}_j &  \mathcal{F}^{R}_j & \mathcal{F}^{I}_{j-1} & \mathbf{0}\\
                        \vdots & \vdots & \mathcal{F}^{R}_j & -\mathcal{F}^{I}_j\\
                        \vdots & \vdots & \mathcal{F}^{I}_j &  \mathcal{F}^{R}_j \\
                             \vdots & \vdots & \vdots & \vdots\\
                        \mathcal{F}^{R}_{\frac{M}{2}} & -\mathcal{F}^{I}_{\frac{M}{2}} &\vdots &\vdots \\
                         \mathcal{F}^{I}_{\frac{M}{2}} & \mathcal{F}^{R}_{\frac{M}{2}} &\vdots &\vdots \\
                         \mathbf{0} & \mathbf{0} &\mathcal{F}^{R}_{\frac{M}{2}} & -\mathcal{F}^{I}_{\frac{M}{2}}\\
                        \mathbf{0} & \mathbf{0} &\mathcal{F}^{I}_{\frac{M}{2}} & \mathcal{F}^{R}_{\frac{M}{2}}
                      \end{array}
\right].
\end{eqnarray}

%


It is seen that from the $(4l-3)$th row to the $(4l)$th row in (\ref{eqhp3}) the  groups $\mathbf{G}_2^{'},~\mathbf{G}_4^{'},~\mathbf{G}_6^{'},~\mathbf{G}_8^{'}$ are all zeros. Similarly, we have that in these four rows, any non-zero linear combination of the vectors in $\mathbf{G}_1^{'}$  over $\Delta \mathcal{A}$ does not belong to the space linearly spanned by all the vectors in the vector groups $\mathbf{G}_2^{'},~\mathbf{G}_4^{'},~\mathbf{G}_6^{'},~\mathbf{G}_8^{'}$.  Additionally, the vector groups $\mathbf{G}_1^{'}, ~\mathbf{G}_3^{'}, ~\mathbf{G}_5^{'}, ~\mathbf{G}_7^{'}$ are orthogonal. Similar to case A1), we have that any non-zero linear combination of the vectors in $\mathbf{G}_1^{'}$  over $\Delta \mathcal{A}$ does not belong to the space linearly spanned by all the vectors in the vector groups $\mathbf{G}_2^{'},~\mathbf{G}_3^{'},\ldots,~\mathbf{G}_8^{'}$.

\item[A3)] $j<l$

In this case, (\ref{eqhp1}) can be expressed as

\begin{eqnarray}\label{eqhp4}
\mathcal{H}^{'}
&=&\left[\begin{array}{cccc}
                       \mathbf{0} & \mathbf{0} & \mathbf{0} & \mathbf{0}\\
                       \vdots & \vdots & \vdots & \vdots\\
                       \mathbf{0}& \mathbf{0} & \vdots& \vdots\\
                        \mathcal{F}^{R}_j &  \mathbf{0} &\mathbf{0} &\mathbf{0}\\
                          \mathbf{0} &  \mathcal{F}^{R}_j &\mathbf{0} &\mathbf{0}\\
                          \vdots & \vdots & \mathcal{F}^{R}_j& \mathbf{0} \\
                            \vdots & \vdots & \mathbf{0} &  \mathcal{F}^{R}_j \\
                         \vdots & \vdots& \vdots& \vdots\\
                          \mathcal{F}^{R}_{l-1} &  \mathbf{0} &\vdots &\vdots\\
                          \mathbf{0} & \mathcal{F}^{R}_{l-1} & \vdots &\vdots\\
                       \mathcal{F}^{R}_{l} &-\mathcal{F}^I_l & \mathcal{F}^{R}_{l-1} &  \mathbf{0} \\
                        \mathcal{F}^{I}_{l} &\mathcal{F}^R_l &  \mathbf{0} & \mathcal{F}^{R}_{l-1} \\
                         \vdots & \vdots &  \mathcal{F}^{R}_{l} &-\mathcal{F}^I_l\\
                          \vdots & \vdots & \mathcal{F}^{I}_{l} &\mathcal{F}^R_l \\
                           \vdots & \vdots &  \vdots & \vdots \\
                        \mathcal{F}^{R}_{\frac{M}{2}} & -\mathcal{F}^{I}_{\frac{M}{2}} &\vdots &\vdots \\
                        \mathcal{F}^{I}_{\frac{M}{2}} & \mathcal{F}^{R}_{\frac{M}{2}} &\vdots &\vdots \\
                        \mathbf{0} & \mathbf{0} &\mathcal{F}^R_{\frac{M}{2}}& -\mathcal{F}^I_{\frac{M}{2}}\\
                         \mathbf{0} & \mathbf{0} &\mathcal{F}^I_{\frac{M}{2}}& \mathcal{F}^R_{\frac{M}{2}}
                      \end{array}
\right].
\end{eqnarray}

%

As for this case, the vector groups $\mathbf{G}_2^{'},~\mathbf{G}_4^{'},~\mathbf{G}_6^{'},~\mathbf{G}_8^{'}$ are all zeros from the $(4j-3)$th row to the $(4j)$th row, and the vector groups $\mathbf{G}_1^{'}, ~\mathbf{G}_3^{'}, ~\mathbf{G}_5^{'}, ~\mathbf{G}_7^{'}$ are orthogonal in (\ref{eqhp4}). Similar to the proof for case A1), we have that any non-zero linear combination of the vectors in $\mathbf{G}_1^{'}$  over $\Delta \mathcal{A}$ does not belong to the space linearly spanned by all the remaining vectors in $\mathbf{G}_2^{'},~\mathbf{G}_3^{'},\ldots,~\mathbf{G}_8^{'}$.
\end{description}


To summarize all the cases A1)-A3),   we then conclude that for $\mathbf{h}_R\neq \mathbf{0}$ and $\mathbf{h}_I\neq \mathbf{0}$ any non-zero linear combination of the vectors in $\mathbf{G}_1^{'}$  over $\Delta \mathcal{A}$ does not belong to the space linearly spanned by all the vectors in the vector groups $\mathbf{G}_2^{'},~\mathbf{G}_3^{'},~\ldots,~\mathbf{G}_8^{'}$.

If the case B) occurs, i.e., $\mathbf{h}_R\neq \mathbf{0}$ and $\mathbf{h}_I= \mathbf{0}$, then (\ref{eqhp1}) can be written as a similar form to (\ref{eqhp4}) by replacing $\mathcal{F}^{I}_{l}$ by $\mathbf{0}$ for all $l$. The proof is the same as that of case A3).

 If the case C) occurs, i.e., $\mathbf{h}_I\neq \mathbf{0}$ and $\mathbf{h}_R= \mathbf{0}$,  then (\ref{eqhp1}) can be written as a similar form to (\ref{eqhp3}) by replacing $\mathcal{F}^{R}_{j}$ by $\mathbf{0}$ for all $j$. The proof is the same as that of case A2).

Therefore, we have proved that for any $\mathbf{h}\neq \mathbf{0}$ any non-zero linear combination of the vectors in $\mathbf{G}_1^{'}$  over $\Delta \mathcal{A}$ does not belong to the space linearly spanned by all the vectors in the vector groups $\mathbf{G}_2^{'},~\mathbf{G}_3^{'},~\ldots,~\mathbf{G}_8^{'}$.

Similarly, we can prove that any non-zero linear combination of the vectors in $\mathbf{G}_p^{'}$  over $\Delta \mathcal{A}$ does not belong to the space linearly spanned by all the vectors in the remaining vector groups, for $p=2, 3, \dots, 8$.

Note that $\mathbf{G}_p^{'}$ is a row permutation of $\mathbf{G}_p$ for $p=1,2,\ldots,8$, respectively. We prove that any non-zero linear combination of the vectors in $\mathbf{G}_p$  over $\Delta \mathcal{A}$ does not belong to the space linearly spanned by all the vectors in the remaining vector groups, for $p=1,2,\ldots,8$.

According to  \emph{proposition 1}, the proof of \emph{Theorem \ref{tpic}} is completed.
\end{proof}

\subsection{Achieving Full Diversity with PIC-SIC Group Decoding for Arbitrary Layers $P$}
In the preceding discussion, the new code $\mathbf{\Phi}_{M,T,P}$ in (\ref{new}) is proved to achieve the full diversity under PIC group decoding when $P=2$ only.
 In the following, we will further show that $\mathbf{\Phi}_{M,T,P}$ with any value $P$ can obtain full diversity  under PIC-SIC group decoding \cite{guoxia}.

\begin{theorem}
 Consider a MIMO system with $M$ transmit antennas and $N$ receive antennas over block fading channels. The  STBC  as described in (\ref{new}) with $P$ diagonal layers is used at the transmitter. The equivalent channel matrix is $\mathcal{H}\in \mathcal{C}^{2TN\times LP}$. If the received signal is decoded using the PIC-SIC group decoding with the grouping scheme $\mathcal{I}={\{\mathcal{I}_1,\cdots,  \mathcal{I}_{4P}\}}$ and with the sequential order, where $\mathcal{I}_p=\{(p-1)M/2+1,\ldots,pM/2\}$ for $p=1,2,\cdots,4P$, i.e., the size of each real group is equal to  $M/2$, then the code $\mathbf{\mathbf{\Phi}}_{M,T,P}$ achieves the full diversity. The code rate of the full-diversity STBC can be up to $M/2$ symbols per channel use.

 The proof is similar to that of \emph{Theorem 1}.  Note that $\mathcal{H}$ for the code $\mathbf{\Phi}_{M,T,P}$ in \emph{Lemma 1} can be written as an alternative form similar to the one in (\ref{eqhp1}) except the expansion of column dimensions. With aid of \emph{proposition 2}, it is simple to follow the  proof for the case of $P=2$ in Section IV-B to prove \emph{Theorem 2}.
The detailed proof is omitted.

\end{theorem}

\section{Simulation Results}
In this section, we present some simulation results for four transmit antennas and four receive antennas. In all simulations, the channel model follows that described in section \ref{model}. In Fig. \ref{fig:1}, four kinds of STBCs are compared: Guo-Xia's code proposed in \cite[Eq. (40)]{guoxia}, $\mathbf{C}_{4,6,2}$ in \cite[Eq. (29)]{wei}, $\mathbf{B}_{4,6,2}$ in \cite[Eq. (37)]{jisit} and the new code $\mathbf{\Phi}_{4,6,2}$ given in (\ref{M4}). Note that all the codes presented in Fig. \ref{fig:1} have the same rate of $4/3$, and 64-QAM constellation is used so that we keep the same bandwidth efficiency of 8 bps/Hz for each code.

Fig. \ref{fig:1} shows the bit error rate (BER) for four codes based on PIC group decoding. Firstly, as expected Guo-Xia's code, $\mathbf{C}_{4,6,2}$, $\mathbf{B}_{4,6,2}$ and the new code $\mathbf{\Phi}_{4,6,2}$ can achieve full diversity at high SNR. Then, one can observe that $\mathbf{\Phi}_{4,6,2}$ has a very similar performance to $\mathbf{B}_{4,6,2}$ and  Guo-Xia's code since we use the same real linear transform  matrix for the case $M=4$. However, compared with $\mathbf{B}_{4,6,2}$ and Guo-Xia's code, the code $\mathbf{\Phi}_{4,6,2}$ further increases the number of PIC groups and allows two real symbols (i.e. single complex symbol) to be decoded in each PIC group without performance loss.

In Fig. \ref{fig:2}, CIOD of rate 1 in \cite[Eq. (85)]{Ra} and  QOSTBC of rate 1 in \cite[Eq. (39)]{quan} with ML decoding are compared with the code
$\mathbf{\Phi}_{4,6,2}$ with PIC group decoding. In order to make a fair performance comparison, the symbols are chosen from a 256QAM signal set for CIOD and QOSTBC, and 64QAM for the code
$\mathbf{\Phi}_{4,6,2}$. Thus, the code $\mathbf{\Phi}_{4,6,2}$ has the same bandwidth efficiency with CIOD and QOSTBC at $8$ bps/Hz. Note that QOSTBC with optimal transformation has a very similar performance to CIOD. Moreover,  one observe that the code $\mathbf{\Phi}_{4,6,2}$ outperforms both CIOD and QOSTBC by $4$ dB. As for this case, the decoding complexity of new code (real symbols pair-wise) is equivalent to that of QOSTBC (real symbols pairwise ML decoding ) and CIOD (single complex symbol ML decoding).

Fig. \ref{fig:3} presents the performance comparison between the code $\mathbf{B}_{4,6,3}$ in \cite{jisit} and the proposed code $\mathbf{\Phi}_{4,6,3}$ with PIC and PIC-SIC group decoding, respectively. Here, 64QAM is used to keep the same bandwidth efficiency of $9$ bps/Hz. It can be observed that $\mathbf{B}_{4,6,3}$ has a very similar performance to $\mathbf{\Phi}_{4,6,3}$ under both PIC and PIC-SIC group decoding. In addition, it is shown that both $\mathbf{B}_{4,6,3}$ and $\mathbf{\Phi}_{4,6,3}$ can achieve full diversity under PIC-SIC group decoding, but lose full diversity when PIC group decoding is employed which is validated by Theorem 1.

\section{Conclusion}
In this paper, we proposed a systematic design of STBC that can achieve full diversity with the PIC group decoding. By coding the real and imaginary parts of the complex symbols vector independently, the proposed code has a reduced PIC group decoding complexity, which is equivalent to a joint decoding of $M/2$ real symbols for $M$ transmit antennas. The full diversity of the proposed STBC with $P$ diagonal layers was proved for PIC group decoding with $P=2$ and PIC-SIC group decoding with any $P$, respectively. It is worthwhile to mention that  for $4$ transmit antennas the code admits real symbol pairwise decoding and the code rate is  $4/3$.  Simulation results show that our proposed code can achieve full diversity with a lower decoding complexity than other existing codes.

\section*{Appendix   - Proof of Lemma 1}

Consider the system described in \emph{Theorem \ref{tpic}} with receive antenna $N=1$. According to the system model given in (\ref{ry}), the matrix form of $\mathbf{y}$ is represented as
\begin{equation}\label{eqn:Y}
\mathbf{y}=\sqrt{\frac{\rho}{\mu}}(\mathbf{A}+j\mathbf{B})(\mathbf{h}_R+j\mathbf{h}_I)+\mathbf{w}.
\end{equation}

With the expansion of (\ref{eqn:Y}), we rewrite $\mathbf{y}=\mathbf{y}_R+j~\mathbf{y}_I$ as a matrix form
\begin{eqnarray}
\mathbf{y}_{R}\label{eqn:yr1}
&=&\sqrt{\frac{\rho}{\mu}}(\mathbf{A}\mathbf{h}_R-\mathbf{B}\mathbf{h}_I)+\mathbf{w}_R,\\
\mathbf{y}_{I}\label{eqn:yi1}
&=&\sqrt{\frac{\rho}{\mu}}(\mathbf{A}\mathbf{h}_I+\mathbf{B}\mathbf{h}_R)+\mathbf{w}_I.
\end{eqnarray}

We substitute the codeword matrices (\ref{AB}) into (\ref{eqn:yr1}) and(\ref{eqn:yi1}). Then, we can obtain
\begin{equation}\label{eqn:yr}
\begin{split}
&\mathbf{y}_R=\left[\begin{array}{c}
\mathbf{y}_{R}^1\\
\mathbf{y}_{R}^2
\end{array}
\right]\\
=
&\sqrt{\frac{\rho}{\mu}}
\left(\left[\begin{array}{cc}
\mathbf{C}_R^{1} & \mathbf{C}_R^{2}\\
-\mathbf{C}^{2}_R & \mathbf{C}^{1}_R\\
\end{array}\right]
\left[\begin{array}{c}
\mathbf{h}_{R}^1\\
\mathbf{h}_{R}^2
\end{array}\right]
-
\left[\begin{array}{cc}
\mathbf{C}_I^{1} & \mathbf{C}_I^{2}\\
\mathbf{C}^{2}_I & -\mathbf{C}^{1}_I\\
\end{array}\right]
\left[\begin{array}{c}
\mathbf{h}_{I}^1\\
\mathbf{h}_{I}^2
\end{array}\right]\right)
+\left[\begin{array}{c}
\mathbf{w}_{R}^1\\
\mathbf{w}_{R}^2
\end{array}\right]\\
=
&\sqrt{\frac{\rho}{\mu}}
\left[\begin{array}{c}
\mathbf{C}_R^{1}\mathbf{h}_{R}^1+\mathbf{C}_R^{2} \mathbf{h}_{R}^2-\mathbf{C}_I^{1}\mathbf{h}_{I}^1-\mathbf{C}_I^{2} \mathbf{h}_{I}^2\\
-\mathbf{C}_R^{2}\mathbf{h}_{R}^1+\mathbf{C}_R^{1} \mathbf{h}_{R}^2-\mathbf{C}_I^{2}\mathbf{h}_{I}^1+\mathbf{C}_I^{1} \mathbf{h}_{I}^2
\end{array}\right]
+\left[\begin{array}{c}
\mathbf{w}_{1,R}\\
\mathbf{w}_{2,R}
\end{array}\right],
\end{split}
\end{equation}

\begin{equation}\label{eqn:yi}
\begin{split}
&\mathbf{y}_I=\left[\begin{array}{c}
\mathbf{y}_{I}^1\\
\mathbf{y}_{I}^2
\end{array}
\right]\\
=
&\sqrt{\frac{\rho}{\mu}}
\left(\left[\begin{array}{cc}
\mathbf{C}_R^{1} & \mathbf{C}_R^{2}\\
-\mathbf{C}^{2}_R & \mathbf{C}^{1}_R
\end{array}\right]
\left[\begin{array}{c}
\mathbf{h}_{I}^1\\
\mathbf{h}_{I}^2
\end{array}\right]
+
\left[\begin{array}{cc}
\mathbf{C}_I^{1} & \mathbf{C}_I^{2}\\
\mathbf{C}^{2}_I & -\mathbf{C}^{1}_I\\
\end{array}\right]
\left[\begin{array}{c}
\mathbf{h}_{R}^1\\
\mathbf{h}_{R}^2
\end{array}\right]\right)
+\left[\begin{array}{c}
\mathbf{w}_{I}^1\\
\mathbf{w}_{I}^2
\end{array}\right]\\
=
&\sqrt{\frac{\rho}{\mu}}
\left[\begin{array}{c}
\mathbf{C}_R^{1}\mathbf{h}_{I}^1+\mathbf{C}_R^{2} \mathbf{h}_{I}^2+\mathbf{C}_I^{1}\mathbf{h}_{R}^1+\mathbf{C}_I^{2} \mathbf{h}_{R}^2\\
-\mathbf{C}_R^{2}\mathbf{h}_{I}^1+\mathbf{C}_R^{1} \mathbf{h}_{I}^2+\mathbf{C}_I^{2}\mathbf{h}_{R}^1-\mathbf{C}_I^{1} \mathbf{h}_{R}^2
\end{array}\right]
+\left[\begin{array}{c}
\mathbf{w}_{1,I}\\
\mathbf{w}_{2,I}
\end{array}\right],
\end{split}
\end{equation}
where $\mathbf{y}^1,\mathbf{y}^2\in \mathbb{C}^{T/2 \times 1}$, $\mathbf{w}^1,\mathbf{w}^2\in \mathbb{C}^{T/2 \times 1}$ and $\mathbf{h}^1,\mathbf{h}^2\in \mathbb{C}^{M/2 \times 1}$. Let $\mathbf{h}^1=[~h_1 ~h_2 ~\ldots ~h_{\frac{M}{2}}~]$ and $\mathbf{h}^2=[~h_{\frac{M}{2}+1} ~h_{\frac{M}{2}+2} ~\ldots ~h_{M}~]$.

Furthermore, according to the code structure in (\ref{C}), (\ref{eqn:yr}) and (\ref{eqn:yi}) can be rewritten as
\begin{equation}\label{eqn:yr2}
\begin{split}
&\mathbf{y}_R=\left[\begin{array}{c}
\mathbf{y}_{R}^1\\
\mathbf{y}_{R}^2
\end{array}
\right]\\
=
&\sqrt{\frac{\rho}{\mu}}\left[\begin{array}{c}
\sum^2_{p=1}\mathbf{C}_{p,R}^{1}\mathbf{h}_{R}^1+\sum^2_{p=1}\mathbf{C}_{p,R}^{2} \mathbf{h}_{R}^2-\sum^2_{p=1}\mathbf{C}_{p,I}^{1}\mathbf{h}_{I}^1-\sum^2_{p=1}\mathbf{C}_{p,I}^{2} \mathbf{h}_{I}^2\\
-\sum^2_{p=1}\mathbf{C}_{p,R}^{2}\mathbf{h}_{R}^1+\sum^2_{p=1}\mathbf{C}_{p,R}^{1} \mathbf{h}_{R}^2-\sum^2_{p=1}\mathbf{C}_{p,I}^{2}\mathbf{h}_{I}^1+\sum^2_{p=1}\mathbf{C}_{p,I}^{1} \mathbf{h}_{I}^2
\end{array}\right]
+\left[\begin{array}{c}
\mathbf{w}_{R}^1\\
\mathbf{w}_{R}^2
\end{array}\right],
\end{split}
\end{equation}
\begin{equation}\label{eqn:yi2}
\begin{split}
&\mathbf{y}_I=\left[\begin{array}{c}
\mathbf{y}_{I}^1\\
\mathbf{y}_{I}^2
\end{array}
\right]\\
=
&\sqrt{\frac{\rho}{\mu}}
\left[\begin{array}{c}
\sum^2_{p=1}\mathbf{C}_{p,R}^{1}\mathbf{h}_{I}^1+\sum^2_{p=1}\mathbf{C}_{p,R}^{2} \mathbf{h}_{I}^2+\sum^2_{p=1}\mathbf{C}_{p,I}^{1}\mathbf{h}_{R}^1+\sum^2_{p=1}\mathbf{C}_{p,I}^{2} \mathbf{h}_{R}^2\\
-\sum^2_{p=1}\mathbf{C}_{p,R}^{2}\mathbf{h}_{I}^1+\sum^2_{p=1}\mathbf{C}_{p,R}^{1} \mathbf{h}_{I}^2+\sum^2_{p=1}\mathbf{C}_{p,I}^{2}\mathbf{h}_{R}^1-\sum^2_{p=1}\mathbf{C}_{p,I}^{1} \mathbf{h}_{R}^2
\end{array}\right]
+\left[\begin{array}{c}
\mathbf{w}_{I}^1\\
\mathbf{w}_{I}^2
\end{array}\right],
\end{split}
\end{equation}
where
\begin{eqnarray}
\mathbf{C}^i_p=&\left[\begin{array}{c}
      \mathbf{0}_{(p-1)\times (M/2)}\\
      \mathrm{diag}(\mathbf{X}^i_p)\\
      \mathbf{0}_{(P-p)\times (M/2)}
      \end{array}
  \right],\, p=1,2;\,\, i=1,2.
  \end{eqnarray}

Equivalently, we have
\begin{equation}\label{eqn:yr2e}
\begin{split}
&\mathbf{y}_R=\left[\begin{array}{c}
\mathbf{y}_{R}^1\\
\mathbf{y}_{R}^2
\end{array}
\right]= \left[\begin{array}{c}
\mathbf{w}_{R}^1\\
\mathbf{w}_{R}^2
\end{array}\right]+\\
&\sqrt{\frac{\rho}{\mu}}\left[\begin{array}{c}
\mathcal{H}^1_{1,R}\mathbf{X}_{1,R}^{1}+\mathcal{H}^1_{2,R}\mathbf{X}_{2,R}^{1} +\mathcal{H}^2_{1,R}\mathbf{X}_{1,R}^{2}+\mathcal{H}^2_{2,R}\mathbf{X}_{2,R}^{2}-
\mathcal{H}^1_{1,I}\mathbf{X}_{1,I}^{1}-\mathcal{H}^1_{2,I}\mathbf{X}_{2,I}^{1}
-\mathcal{H}^2_{1,I}\mathbf{X}_{1,I}^{2}-\mathcal{H}^2_{2,I}\mathbf{X}_{2,I}^{2}\\
-\mathcal{H}^1_{1,R}\mathbf{X}_{1,R}^{2}-\mathcal{H}^1_{2,R}\mathbf{X}_{2,R}^{2}+\mathcal{H}^2_{1,R}\mathbf{X}_{1,R}^{1}+\mathcal{H}^2_{2,R}\mathbf{X}_{2,R}^{1}-\mathcal{H}^1_{1,I}\mathbf{X}_{1,I}^{2}-\mathcal{H}^1_{2,I}\mathbf{X}_{2,I}^{2}+\mathcal{H}^2_{1,I}\mathbf{X}_{1,I}^{1}+\mathcal{H}^2_{2,I}\mathbf{X}_{2,I}^{1}
\end{array}\right],
\end{split}
\end{equation}
\begin{equation}\label{eqn:yi2e}
\begin{split}
&\mathbf{y}_I=\left[\begin{array}{c}
\mathbf{y}_{I}^1\\
\mathbf{y}_{I}^2
\end{array}
\right]=\left[\begin{array}{c}
\mathbf{w}_{I}^1\\
\mathbf{w}_{I}^2
\end{array}\right]+\\
&\sqrt{\frac{\rho}{\mu}}\left[\begin{array}{c}
\mathcal{H}^1_{1,I}\mathbf{X}_{1,R}^{1}+\mathcal{H}^1_{2,I}\mathbf{X}_{2,R}^{1} +\mathcal{H}^2_{1,I}\mathbf{X}_{1,R}^{2}+\mathcal{H}^2_{2,I}\mathbf{X}_{2,R}^{2}+
\mathcal{H}^1_{1,R}\mathbf{X}_{1,I}^{1}+\mathcal{H}^1_{2,R}\mathbf{X}_{2,I}^{1}
+\mathcal{H}^2_{1,R}\mathbf{X}_{1,I}^{2}+\mathcal{H}^2_{2,R}\mathbf{X}_{2,I}^{2}\\
-\mathcal{H}^1_{1,I}\mathbf{X}_{1,R}^{2}-\mathcal{H}^1_{2,I}\mathbf{X}_{2,R}^{2}+\mathcal{H}^2_{1,I}\mathbf{X}_{1,R}^{1}+\mathcal{H}^2_{2,I}\mathbf{X}_{2,R}^{1}+\mathcal{H}^1_{1,R}\mathbf{X}_{1,I}^{2}+\mathcal{H}^1_{2,R}\mathbf{X}_{2,I}^{2}-\mathcal{H}^2_{1,R}\mathbf{X}_{1,I}^{1}-\mathcal{H}^2_{2,R}\mathbf{X}_{2,I}^{1}
\end{array}\right]
\end{split}
\end{equation}
where
\begin{eqnarray}
\mathcal{H}_{p,R}^i=\left[\begin{array}{c}
      \mathbf{0}_{(p-1)\times (M/2)}\\
      \mathrm{diag}(\mathbf{h}_R^{i})\\
      \mathbf{0}_{(2-p)\times (M/2)}
      \end{array}
  \right], {\rm{and}~}
\mathcal{H}_{p,I}^i=\left[\begin{array}{c}
      \mathbf{0}_{(p-1)\times (M/2)}\\
      \mathrm{diag}(\mathbf{h}_I^{i})\\
      \mathbf{0}_{(2-p)\times (M/2)}
      \end{array}
  \right],
\end{eqnarray}
for $p=1,2$ and $i=1,2$.

Next, we gather the equations $\mathbf{y}_{R}^1$, $\mathbf{y}_{R}^2$, $\mathbf{y}_{I}^1$ and $\mathbf{y}_{I}^2$ to form a real system as follows
\begin{eqnarray}\label{eqn:ye1}\nonumber
\left[\begin{array}{c}
\mathbf{y}_{R}^1\\
\mathbf{y}_{R}^2\\
\mathbf{y}_{I}^1\\
\mathbf{y}_{I}^2
   \end{array}
  \right]&=&
 \sqrt{\frac{\rho}{\mu}}
\left[\begin{array}{cccccccc}
  \mathcal{H}^1_{1,R}&\mathcal{H}^1_{2,R}&\mathcal{H}^2_{1,R}&\mathcal{H}^2_{2,R}&-\mathcal{H}^1_{1,I}&-\mathcal{H}^1_{2,I}&-\mathcal{H}^2_{1,I}&-\mathcal{H}^2_{2,I}\\
\mathcal{H}^2_{1,R}&\mathcal{H}^2_{2,R}&-\mathcal{H}^1_{1,R}&-\mathcal{H}^1_{2,R}&\mathcal{H}^2_{1,I}&\mathcal{H}^2_{2,I}&-\mathcal{H}^1_{1,I}&-\mathcal{H}^1_{2,I}\\
   \mathcal{H}^1_{1,I}&\mathcal{H}^1_{2,I}&\mathcal{H}^2_{1,I}&\mathcal{H}^2_{2,I}&\mathcal{H}^1_{1,R}&\mathcal{H}^1_{2,R}&\mathcal{H}^2_{1,R}&\mathcal{H}^2_{2,R}\\
  \mathcal{H}^2_{1,I}&\mathcal{H}^2_{2,I}&-\mathcal{H}^1_{1,I}&-\mathcal{H}^1_{2,I}&-\mathcal{H}^2_{1,R}&-\mathcal{H}^2_{2,R}&\mathcal{H}^1_{1,R}&\mathcal{H}^1_{2,R}
  \end{array}
  \right]
    \left[\begin{array}{c}
    \mathbf{X}^1_{1,R}\\\mathbf{X}^1_{2,R}\\ \mathbf{X}^2_{1,R} \\\mathbf{X}^2_{2,R} \\ \mathbf{X}^1_{1,I} \\ \mathbf{X}^1_{2,I} \\\mathbf{X}^2_{1,I} \\\mathbf{X}^2_{2,I}
    \end{array}
  \right]\\
  &+&\left[\begin{array}{c}
\mathbf{w}_{R}^1\\
\mathbf{w}_{R}^2\\
\mathbf{w}_{I}^1\\
\mathbf{w}_{I}^2
\end{array}\right].
\end{eqnarray}


In order to obtain the equivalent signal model in (\ref{eqyr}), using (\ref{roxir}) and (\ref{roxii}) we can rewrite (\ref{eqn:ye1}) as
\begin{eqnarray}\label{eqn:ye2}
\left[\begin{array}{c}
\mathbf{y}_{R}^1\\
\mathbf{y}_{R}^2\\
\mathbf{y}_{I}^1\\
\mathbf{y}_{I}^2   \end{array}
  \right]=
 \sqrt{\frac{\rho}{\mu}}\mathcal{H}
    \left[\begin{array}{c}
  \mathbf{s}^1_{1,R}\\ \mathbf{s}^1_{2,R}\\ \mathbf{s}^2_{1,R} \\ \mathbf{s}^2_{2,R}\\ \mathbf{s}^1_{1,I}\\ \mathbf{s}^1_{2,I}\\ \mathbf{s}^2_{1,I}\\ \mathbf{s}^2_{2,I}
    \end{array}
  \right]
  +\left[\begin{array}{c}
\mathbf{w}_{R}^1\\
\mathbf{w}_{R}^2\\
\mathbf{w}_{I}^1\\
\mathbf{w}_{I}^2
\end{array}\right],
\end{eqnarray}
where the equivalent real channel matrix $\mathcal{H} \in \mathbb{R}^{2T \times 4M}$  is given by
\begin{eqnarray}\label{eqn:ye}
\mathcal{H}&=& \left[\begin{array}{cccccccc}\nonumber
\mathcal{H}^1_{1,R}\mathbf{\Theta}_{A}&\mathcal{H}^1_{2,R}\mathbf{\Theta}_{A}&\mathcal{H}^2_{1,R}\mathbf{\Theta}_{A}&\mathcal{H}^2_{2,R}\mathbf{\Theta}_{A}&-\mathcal{H}^1_{1,I}\mathbf{\Theta}_{B}&-\mathcal{H}^1_{2,I}\mathbf{\Theta}_{B}&-\mathcal{H}^2_{1,I}\mathbf{\Theta}_{B}&-\mathcal{H}^2_{2,I}\mathbf{\Theta}_{B}\\
\mathcal{H}^2_{1,R}\mathbf{\Theta}_{A}&\mathcal{H}^2_{2,R}\mathbf{\Theta}_{A}&-\mathcal{H}^1_{1,R}\mathbf{\Theta}_{A}&-\mathcal{H}^1_{2,R}\mathbf{\Theta}_{A}&\mathcal{H}^2_{1,I}\mathbf{\Theta}_{B}&\mathcal{H}^2_{2,I}\mathbf{\Theta}_{B}&-\mathcal{H}^1_{1,I}\mathbf{\Theta}_{B}&-\mathcal{H}^1_{2,I}\mathbf{\Theta}_{B}\\
   \mathcal{H}^1_{1,I}\mathbf{\Theta}_{A}&\mathcal{H}^1_{2,I}\mathbf{\Theta}_{A}&\mathcal{H}^2_{1,I}\mathbf{\Theta}_{A}&\mathcal{H}^2_{2,I}\mathbf{\Theta}_{A}&\mathcal{H}^1_{1,R}\mathbf{\Theta}_{B}&\mathcal{H}^1_{2,R}\mathbf{\Theta}_{B}&\mathcal{H}^2_{1,R}\mathbf{\Theta}_{B}&\mathcal{H}^2_{2,R}\mathbf{\Theta}_{B}\\
  \mathcal{H}^2_{1,I}\mathbf{\Theta}_{A}&\mathcal{H}^2_{2,I}\mathbf{\Theta}_{A}&-\mathcal{H}^1_{1,I}\mathbf{\Theta}_{A}&-\mathcal{H}^1_{2,I}\mathbf{\Theta}_{A}&-\mathcal{H}^2_{1,R}\mathbf{\Theta}_{B}&-\mathcal{H}^2_{2,R}\mathbf{\Theta}_{B}&\mathcal{H}^1_{1,R}\mathbf{\Theta}_{B}&\mathcal{H}^1_{2,R}\mathbf{\Theta}_{B}  \end{array}
  \right]\\
  &=&\left[\begin{array}{cccccccc}\label{pic}
  \mathbf{G}_1 &\mathbf{G}_2 &\ldots & \mathbf{G}_8
  \end{array}
\right].
  \end{eqnarray}

  According to \emph{Definition \ref{orde}}, we obtain that the groups $\mathbf{G}_1, ~\mathbf{G}_3, ~\mathbf{G}_5, ~\mathbf{G}_7$ are orthogonal, and the groups $\mathbf{G}_2, ~\mathbf{G}_4, ~\mathbf{G}_6, ~\mathbf{G}_8$ are orthogonal as well.

\section*{Acknowledgment}
The authors would like to thank Tianyi Xu for his reading and comments on this manuscript.

\newpage
\begin{table}[t]\label{table:picM4}
\begin{center}\caption{Comparison in PIC Group Decoding Complexity}
\renewcommand{\arraystretch}{1.5}
\centering
\begin{tabular}{c c c c c}
\hline \hline
\bfseries Codes   & \bfseries  Groups & \bfseries  Symbols/Group & \bfseries Decoding Complexity\\
\hline
$\mathbf{C}_{4,5,2}/\mathbf{C}_{4,6,2}$ \cite{wei}   & $2$ & $4$ (Complex) & $2{|\mathcal{A}|^4}$\\
$\mathbf{\mathbf{B}}_{4,6,2}$\cite{jisit}    & $4$ & $2$ (Complex) & $4{|\mathcal{A}|^2}$\\
$\mathbf{\mathbf{\Phi}}_{4,6,2}$   & $8$ & $2$ (Real) & $8{|\mathcal{A}|}$\\
\hline\hline
\end{tabular}
\end{center}
\end{table}

\begin{figure}[t!]
    \begin{center}
        \includegraphics[width=1\columnwidth]{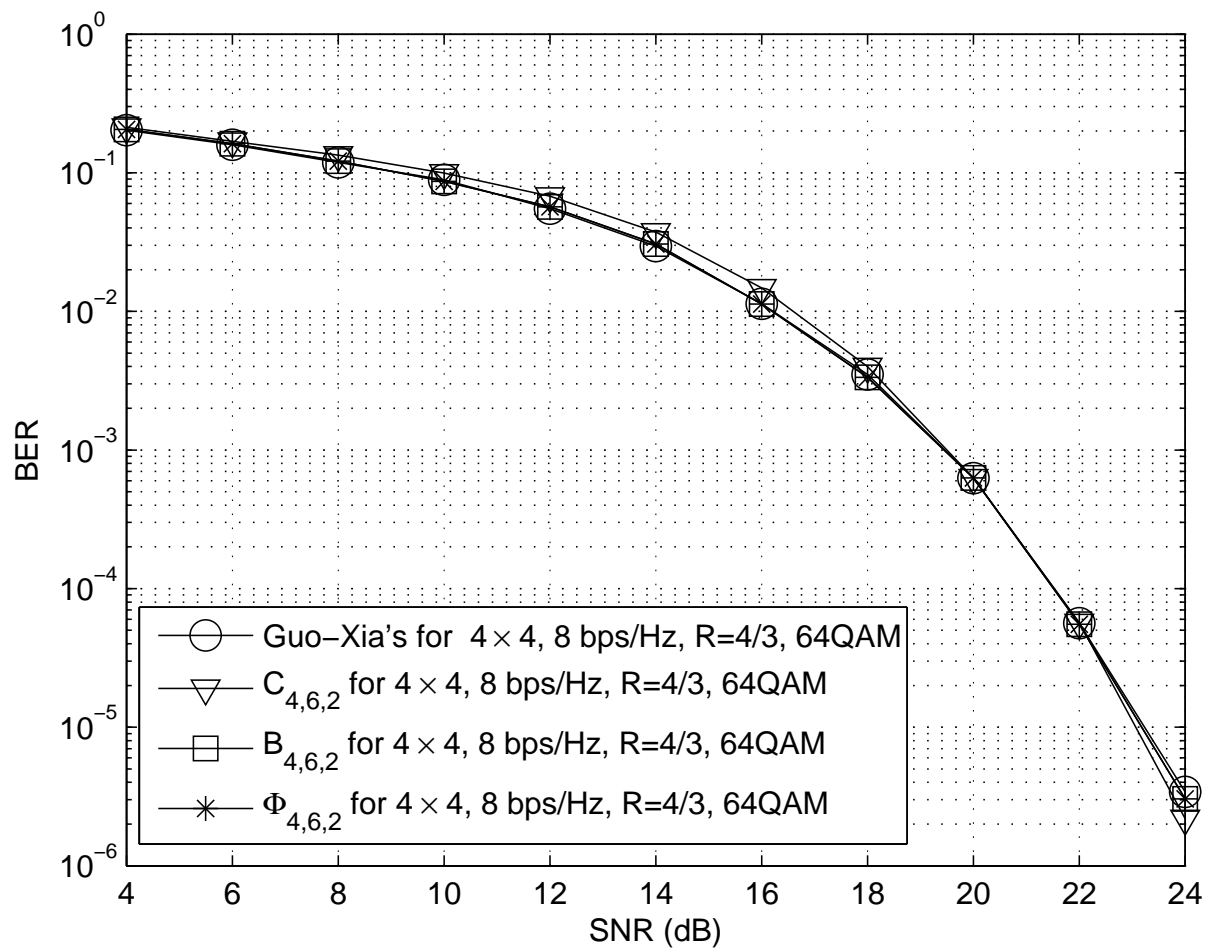}
        \caption{BER performance of various codes with PIC group decoding with 4 transmit antennas and 4 receive antennas.}
        \label{fig:1}
    \end{center}
\end{figure}
\newpage
\begin{figure}[t!]
    \begin{center}
        \includegraphics[width=1\columnwidth]{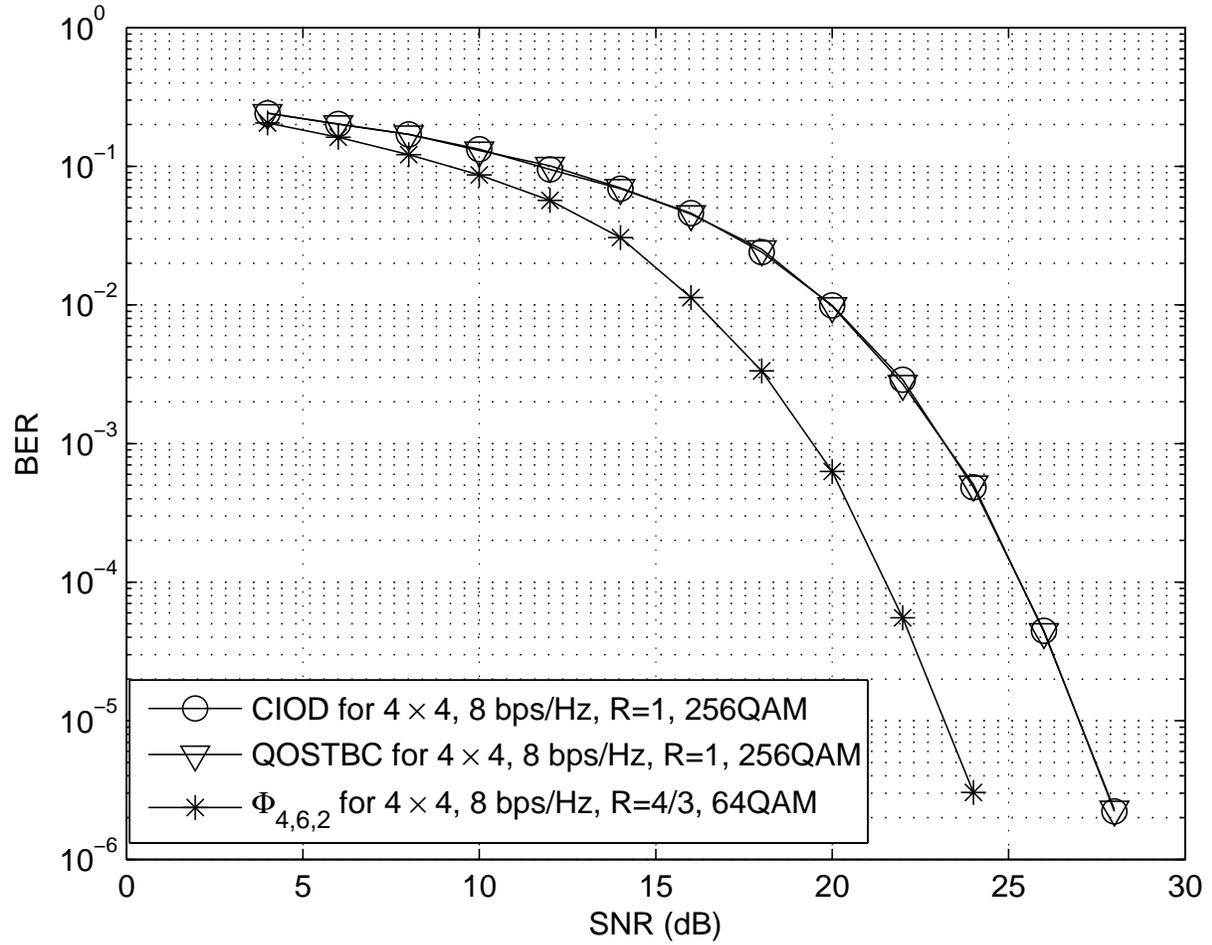}
        \caption{BER performance of various codes with 4 transmit antennas and 4 receive antennas at $8$ bps/Hz.}
        \label{fig:2}
    \end{center}
\end{figure}
\newpage
\begin{figure}[t!]
    \begin{center}
        \includegraphics[width=1\columnwidth]{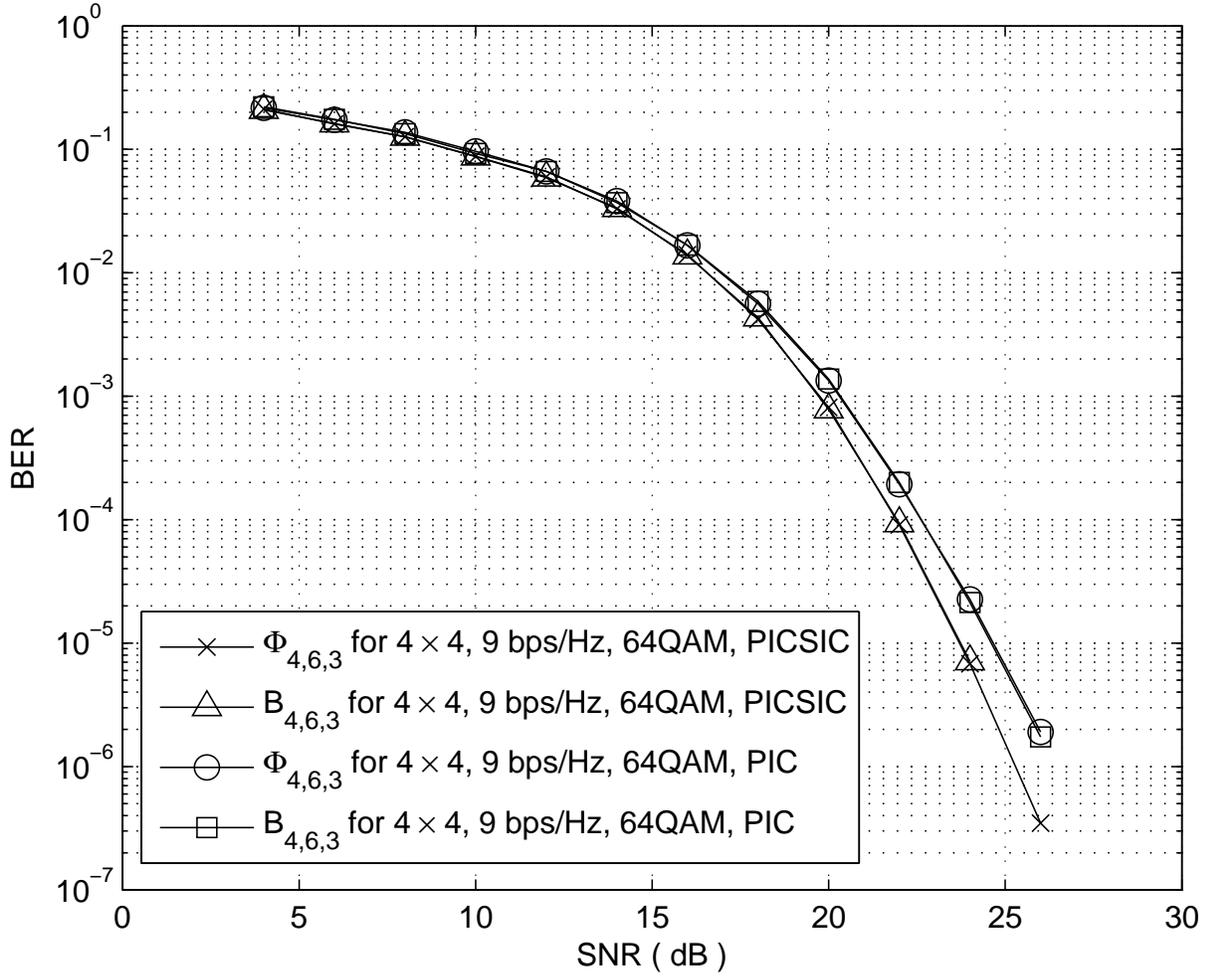}
        \caption{BER performance of various codes with 4 transmit antennas and 4 receive antennas at $9$ bps/Hz with PIC and PIC-SIC group decoding.}
        \label{fig:3}
    \end{center}
\end{figure}
\end{document}